\documentclass[11pt]{article}
\usepackage{amsmath}
\usepackage{amssymb}
\usepackage{amsthm}
\usepackage[pdftex]{graphicx}

\def\clift#1{#1^{\scriptscriptstyle{\mathrm{C}}}}

\def\vlift#1{#1^{\scriptscriptstyle{\mathrm{V}}}}

\def\C{\mathcal{C}}
\def\D{\mathcal{D}}
\def\fpd#1#2{{\displaystyle\frac{\partial #1}{\partial #2}}}

\def\lie#1{{\mathcal{L}}_{#1}}

\def\R{\mathbb{R}}

\def\vf#1{\frac{\partial}{\partial #1}}
\def\vectorfields#1{{\cal X}(#1)}

\def\im{\mathop{\mathrm{im}}}

\font\frak=eufm10 scaled\magstep1
\def\goth #1{\hbox{{\frak #1}}}
\def\g{\goth{g}}

\def\conn#1#2#3{\setbox1=\hbox{$\scriptstyle{#2}{#3}$}%
\setbox2=\hbox to\wd1{$\hfil\scriptstyle{#1}\hfil$}
\Gamma^{\!\box2}_{\!\box1}}
\def\barconn#1#2#3{\setbox1=\hbox{$\scriptstyle{#2}{#3}$}%
\setbox2=\hbox to\wd1{$\hfil\scriptstyle{#1}\hfil$}
\check{\Gamma}^{\!\box2}_{\!\box1}}

\def\onehalf{{\textstyle\frac12}}

\def\hook{{\mathchoice
{\vrule height 0pt depth 0.4pt width 3pt \vrule height 5pt depth 0.4pt
\kern 3pt}
{\vrule height 0pt depth 0.4pt width 3pt  \vrule height 5pt depth
0.4pt\kern 3pt}
{\vrule height 0pt depth 0.2pt width 1.5pt  \vrule height 3pt depth
0.2pt width 0.2pt \kern 1pt}
{\vrule height 0pt depth 0.2pt width 1.5pt  \vrule height 3pt depth
0.2pt width 0.2pt \kern 1pt} }}

\newtheorem{thm}{Theorem}

\newtheorem{prop}{Proposition}
\newtheorem{cor}{Corollary}

\begin{document}

\title{Anholonomic frames in constrained dynamics}

\author{M.\ Crampin and T.\ Mestdag\\[2mm]
{\small Department of Mathematical Physics and Astronomy, Ghent University}\\
{\small Krijgslaan 281, B-9000 Ghent, Belgium}}

\date{}

\maketitle

{\small {\bf Abstract.} We demonstrate the usefulness of anholonomic
frames in the contexts of nonholonomic and vakonomic
systems. We take a consistently
differential-geometric approach. As an application, we
investigate the conditions under which the dynamics of the two systems
will be consistent.  A few illustrative examples confirm the results.
\\[2mm] {\bf Mathematics Subject Classification (2000).} 34A26,
37J60, 70G45, 70G75, 70H03.\\[2mm] {\bf Keywords.} Lagrangian system,
nonholonomic constraints, anholonomic frames, vakonomic dynamics}

\section{Introduction}

A frame on an $n$-dimensional differentiable manifold $Q$ is a local
basis of vector fields $\{X_i\}$, $i=1,2,\ldots,n$.  A frame
consisting of coordinate vector fields $\partial/\partial q^i$ is
called, naturally enough, a coordinate frame, or more obscurely a
holonomic one; a frame which is not holonomic, or not known to be so,
is said to be anholonomic.  The term dates back at least as far as
Schouten's {\em Ricci-Calculus\/} \cite{Sch}.  Why one says
anholonomic for frames but nonholonomic for constraints in dynamics,
with closely related meanings, is anybody's guess.

In order to test whether or not a frame $\{X_i\}$ is a coordinate
frame (that is, whether coordinates may be found such that
$X_i=\partial/\partial q^i$) one computes the brackets $[X_i,X_j]$.
These vanish for a coordinate frame; for an anholonomic frame we may
write
\[
[X_i,X_j]=R^k_{ij}X_k
\]
for some locally defined functions $R^k_{ij}$ on $Q$. These functions
are collectively called by Schouten the object of anholonomity,
because of their role in determining whether or not the frame is
truly anholonomic.

Anholonomic frames can be very useful in certain situations in
dynamics because they can be adapted to geometrical requirements in a
way that may be impossible with a coordinate frame.  We have exploited
this possibility in relation to systems with symmetry in previous
publications \cite{Routh,releq,Lie,red,invlag}. In this paper, we
promote the use of anholonomic frames in the study of dynamical
systems subject to nonholonomic linear constraints.

Our paper is concerned entirely with dynamical systems of finitely
many degrees of freedom which are capable, broadly speaking, of a
Lagrangian formulation; and it studies the systems of interest by
uncompromisingly differential-geometric means.  One instance of this
is that, for us, the dynamics is always represented by a vector field;
the possible motions of a dynamical system are given by the integral
curves of the dynamical vector field, and the dynamical differential
equations are those that determine the integral curves of the
dynamical vector field.  When one is dealing with a Lagrangian system
(without constraints, for simplicity) one has therefore to interpret
the Euler-Lagrange equations as implicitly defining a vector field:\
this will in fact be a vector field of special type --- a so-called
second-order differential equation field --- on the tangent bundle
$TQ$ of the configuration space $Q$.  Special techniques have been
developed for the study of the differential geometry of tangent
bundles and second-order differential equation fields, which have been
shown to be useful in many applications, such as the inverse problem
of Lagrangian mechanics and the study of qualitative features of
systems of second-order differential equations.

Our insistence on always thinking of the dynamics as a vector field of
second-order type on $TQ$ is one distinctive feature of our approach.
A second is our whole-hearted use of anholonomic frames.  For example,
we reformulate the Euler-Lagrange equations (which of course are
normally expressed in terms of coordinates) in a frame-dependent but
entirely coordinate-free manner.  To do so we first show how to
lift an anholonomic frame on $Q$ to one on $TQ$. In keeping with our
differential-geometric approach we do not base our discussion of
Euler-Lagrange equations on a variational principle, but instead
derive the equations ultimately by the method of the Cartan form (see \cite{CP}, for example, for a general description, and Section~2.3 for one adapted to our needs in this paper).

The concept of quasi-velocities has a very natural place in our theory
--- indeed, we would claim that the only way to understand
quasi-velocities properly is to view them from the perspective of the
theory of anholonomic frames.  The thoroughly misleading --- indeed,
incoherent --- concept of quasi-coordinates to be found in some
classical textbooks on mechanics is thereby avoided.  It is argued in
a recent paper by Bloch {\em et al.}\ \cite{BMZ} that although the
benefits of the use of quasi-velocities have always been beyond any
doubt, the mathematical foundations of the theory in textbooks such as
\cite{Greenwood} and \cite{Neimark} have not always been built up as
rigourously as they should.  We agree; and in fact we take a more
radical approach to the problem than do the authors of \cite{BMZ}.
For example, the nearest these authors come to a frame-based version
of the Euler-Lagrange equations is Hamel's equations.  We on the other
hand view Hamel's equations, however useful they may be in practical
applications, as just a half-way house --- partly frame- and partly
coordinate-based.  We show in our paper how to derive Hamel's
equations, but we make no use of them.

The differential-geometric machinery necessary for our approach is
described in Section~2.1 below. Our version of the Euler-Lagrange equations is obtained in Proposition~1 in Section~2.2. In Section~2.3 we show how to obtain these equations by means of the Cartan forms.

There has been a long discussion in the literature about two distinct
approaches to dealing with dynamical systems subject to nonholonomic
constraints, which are called respectively nonholonomic and vakonomic
mechanics.  According to \cite{Bloch,Cortes,Rumiantsev} the discussion
dates back to the end of the nineteenth century, with Korteweg,
H\"older, Voronec and Suslov among its more famous participants.
Nonholonomic mechanics is the classical method for deriving equations
of motion for systems with constraints and their associated reaction
forces from the Lagrange-d'Alembert principle, while the equations for
vakonomic systems (sometimes also termed variational nonholonomic
systems) follow from a variational principle where one looks for
extremals in the class of curves which satisfy the constraints.
Experiments such as the one described in \cite{LewisMurray} suggest
that the nonholonomic equations of motion are the correct ones to use
for mechanical systems with nonholonomic constraints.  The vakonomic
theory, on the other hand, finds applications in such fields as
economics and LC-circuits (see for example \cite{CDMM}).

In Section 3 we show in Proposition~2 how our methods may be used to formulate a
version of the Lagrange-d'Alembert principle and derive the
nonholonomic dynamics for a system subject to nonintegrable linear
constraints, provided it is regular in an appropriate sense.

In Section 4 we turn our attention to vakonomic systems.  We derive the equations satisfied by the dynamical
vector fields of such systems. In Propositions~3 and 4 we give a version of these equations, when the multipliers are required to lie in the image of a certain section. Section 5
concentrates in much detail on the question of when the nonholonomic
problem is consistent with the vakonomic one.  This question has of
course been investigated by many authors in the past (see for example
\cite{CF,CDMM,unified,Favretti,IMMS,Rumiantsev,Zampieri} --- this list is
by no means exhaustive). By carefully analyzing the dynamical vector
fields involved we show that in fact two notions of consistency need
to be distinguished. The corresponding consistency conditions can be found in Theorem~2 and Corollary~1, repectively. Inspired by a recent paper of Fernandez and
Bloch \cite{FandB} we examine the existence of a so-called
`variational Lagrangian' and we clarify its role with regard to the
consistency conditions (Theorem~1 and Corollary~2).  We considerably generalize some of the
results of \cite{FandB}, stated there for Abelian Chaplygin systems
only, by showing that they apply also to non-Abelian Chaplygin
systems (in the sense of e.g.\ \cite{Koiller}) and even in some cases
to nonholonomic systems with linear constraints in general (see Propositions~5 and 6).  Our
discussion of Chaplygin systems is to be found in Section 6. As in earlier sections, the ease with which we obtain the results relies on the choice of an appropriate frame. In the case of a Chaplygin system the frame is adapted to the situation of a system with symmetry. In the literature one can, of course, also find other approaches to (nonholonomic and vakonomic) systems with symmetry and their reduction, e.g.\ using the theory of Lie algebroids \cite{algebroid,IMMS,ML}; but the benefits of our methods are of course not available there.

The paper ends with some illustrative examples, namely the class of
nonholonomic systems discussed recently in \cite{BFM}, and the
two-wheeled carriage.  We give new and illuminating derivations of
some known results.  We also correct the errors which have crept into
some of the accounts of the two-wheeled carriage in the
literature.

\section{Anholonomic frames, quasi-velocities and the Euler-\\
Lagrange equations}

In this section we describe the basic constructions associated with
anholonomic frames, leading to an appropriate formulation of the
Euler-Lagrange equations, in preparation for the discussion of
nonholonomic systems which commences in the following section.

\subsection{Some aspects of tangent bundle geometry}

We shall be concerned with second-order dynamics, that is to say, with
dynamical systems represented as vector fields of second-order type on
velocity phase space, which is the tangent bundle $TQ$ of
configuration space $Q$.  We denote the space of sections of
$TQ\to Q$, that is the space of vector fields on $Q$, by
$\vectorfields{Q}$; it is a $C^\infty(Q)$-module, where $C^\infty(Q)$
denotes the ring of smooth real-valued functions on $Q$.  (Actually all
considerations in this paper are local rather than global, but we shall
not continually draw attention to this fact.)  Though in principle we
prefer to manage without coordinates, we shall want to use them
sometimes, especially in this section:\ we point out that we
denote the generalized velocities, in other words the fibre
coordinates on $TQ$ naturally associated with coordinates $q^i$ on
$Q$, by $u^i$, reserving $\dot{q}^i$ for the actual derivative of
$q^i$ along a curve.  Our policy in this section is to give informal
but intrinsic definitions of objects and constructions so that it will
be clear that they are well-defined, but to supplement these with
coordinate expressions for security:\ this is why we make more use of
coordinates here than later.

We shall work with anholonomic frames on $Q$.  Our first task is to
show how to lift an anholonomic frame from $Q$ to $TQ$, to give an
anholonomic frame there.  Since the dimension of $TQ$ is $2n$ we need
to double the number of vector fields in the frame.  Fortunately there
are two canonical ways of lifting a vector field from $Q$ to $TQ$, the
so-called complete and vertical lifts; we apply them both to each
member of the frame on $Q$.  These constructions are described in
detail in several texts, including \cite{CP}; we give a brief account
below for the convenience of the reader.

Let $Z$ be any (locally defined) vector field on $Q$, with flow $\varphi_t$. By taking the tangent map to $\varphi_t$ for each $t$ we obtain a flow on $TQ$; the corresponding vector field on $TQ$ is called the complete
lift (sometimes the tangent lift) of $Z$ and is denoted by $\clift{Z}$. In terms of coordinates
\[
\clift Z=Z^i\vf{q^i}+u^j\fpd{Z^i}{q^j}\vf{u^i},\quad
\mbox{if }Z=Z^i\vf{q^i}.
\]
The second canonical way of lifting the vector field $Z$ from $Q$ to
$TQ$ yields its vertical lift $\vlift{Z}$:\ this is tangent to the
fibres of the projection $\tau: TQ \to Q$, and on the fibre over $q\in
Q$ it coincides with the constant vector field $Z_q$.  In coordinates
\[
\vlift{Z}=Z^i\vf{u^i}.
\]
We have $\tau_*\clift{Z}=Z$, which is to say that $\clift{Z}$ is
projectable and projects onto $Z$, while $\tau_*\vlift{Z}=0$ since
$\vlift{Z}$ is vertical.  The brackets of complete and vertical lifts
of vector fields $Y$ and $Z$ are
\[
[\clift Y,\clift Z] = \clift{[Y,Z]},\quad
[\clift Y,\vlift Z] = \vlift{[Y,Z]},\quad
[\vlift Y,\vlift Z]=0.
\]
Note that although the map $Z\mapsto\vlift{Z}$ is
$C^\infty(Q)$-linear, the map $Z\mapsto\clift{Z}$ is only
$\R$-linear:\ in fact for $f\in C^\infty(Q)$
\[
\clift{(fZ)}=f\clift{Z}+\dot{f}\vlift{Z}
\]
where $\dot{f}$ is the so-called total derivative of $f$, a function
on $TQ$ defined by $\dot{f}(q,u)=u(f)$ and given in coordinates by
\[
\dot{f}=u^i\fpd{f}{q^i}.
\]
This fact will have an important role to play shortly.  (The first
term on the right of the equation above for $\clift{(fZ)}$ should
strictly speaking be $(\tau^*f)\clift{Z}$, but we shall not bother to
distinguish notationally between a function on $Q$ and its pull-back
to $TQ$.)

Given an anholonomic frame $\{X_i\}$ on $Q$ we can construct from it
the anholonomic frame $\{\clift{X_i},\vlift{X_i}\}$ on $TQ$.  This
lifted frame will play a central role in our account of dynamics on
$TQ$.

With any local frame $\{X_i\}$ there is associated its dual coframe
$\{\theta^i\}$, consisting of locally defined 1-forms such that
$\theta^i(X_j)=\delta^i_j$.  Now a 1-form $\theta$ on $Q$ defines a
function say $v_\theta$ on $TQ$, linear on the fibres, by
$v_\theta(q,u)=\theta_q(u)$; if $\theta=\theta_idq^i$ then
$v_\theta=\theta_iu^i$.  The $n$ functions $v_{\theta_i}$ associated
in this way with a frame via the dual coframe are called the
quasi-velocities of the frame; we denote them by $v^i$ for
simplicity. Another way of defining the quasi-velocities is to say
that $v^i(q,u)$ is just the $i$th component of the vector $u\in T_qQ$
when it is expressed in terms of the frame at $q$.

We shall need expressions for the derivatives of the
quasi-velocities by complete and vertical lifts, and in particular
their derivatives along the members of the lifted frame
$\{\clift{X_i},\vlift{X_i}\}$.  To find them, the following two
formulas are indispensable:
\[
\clift{Z}(v_\theta)=v_{\lie{Z}\theta},\quad
\vlift{Z}(v_\theta)=\theta(Z).
\]
From the second of these, $\vlift{Z}(v^i)=Z^i$, where (for any vector
field) $Z^i=\theta^i(Z)$ is the $i$th component of $Z$ with respect to
the frame.  From the first, $\clift{Z}(v^i)=-[Z,X_j]^iv^j$.  Thus we
can express $\clift{X_i}(v^j)$ in terms of the object of
anholonomity:\ if $[X_i,X_j]=R^k_{ij}X_k$ then
\[
\clift{X_i}(v^j)=-R^j_{ik}v^k.
\]
On the other hand
\[
\vlift{X_i}(v^j)=\delta^j_i,
\]
from the second formula displayed above.

\subsection{The Euler-Lagrange field}

We now turn our attention to dynamics.  We deal with second-order
systems, which means that we expect to be able to write the dynamical
equations as second-order ordinary differential equations of the form
$\ddot{q}^i=F^i(q,\dot{q})$.  As we explained in the introduction, we
follow the time-honoured principle of differential geometry that a
suitable system of ordinary differential equations should be replaced
by the vector field whose integral curves are its solutions.  A system
of equations of the form $\ddot{q}^i=F^i(q,\dot{q})$ corresponds in
this way to a vector field $\Gamma$ on $TQ$, but one of a special
character:\ we must have $\tau_{*(q,u)}\Gamma=u$, which is to say that
$\Gamma$ must take the form
\[
\Gamma =u^i\vf{q^i}+F^i\vf{u^i}
\]
in coordinates, for then the equations for its integral curves are
\[
\dot{q}^i=u^i,\quad \dot{u}^i=F^i(q,u),
\]
as required. Such a vector field is, naturally enough, called a
second-order differential equation field.

One useful property of second-order differential equation fields,
which is most easily checked by a coordinate calculation, is that for
any vector field $Z$ on $Q$, $[\clift{Z},\Gamma]$ is
vertical.

Suppose we have an anholonomic frame $\{X_i\}$ on $Q$; we can then
express any vector field on $TQ$ in terms of the lifted frame
$\{\clift{X_i},\vlift{X_i}\}$. Since $\tau_*\clift{X_i}=X_i$, the
condition for a vector field to be a second-order differential
equation field, when expressed in these terms, is that it should take
the form
\[
v^i\clift{X_i}+\Gamma^i\vlift{X_i}
\]
for some functions $\Gamma^i$ on $TQ$, where the coefficients $v^i$
are the quasi-velocities of the frame.

Note that for any function $f$ on $Q$ and any second-order
differential equation field $\Gamma$ we have $\Gamma(f)=\dot{f}$.
Recall that the total derivative $\dot{f}$ also appears in the expression
for $\clift{(fZ)}$, which suggests the possibility of constructive
cancellation. Let $\Gamma$ be any second-order differential equation
field and $L$ any function on $TQ$, and consider the map
$\varepsilon:\vectorfields{Q}\to C^\infty(TQ)$ given by
\[
\varepsilon(Z)=\Gamma(\vlift{Z}(L))-\clift{Z}(L).
\]
Then $\varepsilon$ is evidently $\R$-linear; it is in fact
$C^\infty(Q)$-linear, as the following calculation shows:
\[
\varepsilon(fZ)=\Gamma(\vlift{fZ}(L))-\clift{(fZ)}(L)
=\dot{f}\vlift{Z}(L)+f\Gamma(\vlift{Z}(L))-f\clift{Z}(L)-\dot{f}\vlift{Z}(L)
=f\varepsilon(Z).
\]
Thus $\varepsilon$ behaves like a 1-form, except that it takes its
values in $C^\infty(TQ)$ rather than $C^\infty(Q)$:\ it is in fact a
1-form along the tangent bundle projection $\tau$.

A regular Lagrangian $L$ determines a second-order dynamical system
via its Euler-Lagrange equations. These are usually written as
\[
\frac{d}{dt}\left(\fpd{L}{\dot{q}^i}\right)-\fpd{L}{q^i}=0.
\]
When $L$ is regular, which means that its Hessian with respect to the
velocities is nonsingular, these equations can be solved for the
$\ddot{q}^i$, or in other words written in the form
$\ddot{q}^i=F^i(q,\dot{q})$. They correspond to a second-order
differential equation field $\Gamma$, therefore. In fact there is a
unique second-order differential equation field $\Gamma$ such that
\[
\Gamma\left(\fpd{L}{u^i}\right)-\fpd{L}{q^i}=0,
\]
and its integral curves are the solutions of the Euler-Lagrange
equations. But in view of the remarks above about $\varepsilon$, we
see that $\Gamma$ must in fact satisfy
\[
\Gamma(\vlift{Z}(L))-\clift{Z}(L)=0
\]
for any vector field $Z$ on $Q$.  Moreover, to determine $\Gamma$,
assuming that it is a second-order differential equation field and
that $L$ is regular, it is enough to require that
\[
\Gamma(\vlift{X_i}(L))-\clift{X_i}(L)=0, \quad i=1,2,\ldots n,
\]
for the vector fields $X_i$ of any frame on $Q$, even an anholonomic
one.  To be explicit, we may take
$\Gamma=v^i\clift{X_i}+\Gamma^i\vlift{X_i}$, when the equation above
becomes
\[
\vlift{X_i}(\vlift{X_j}(L))\Gamma^j=
\clift{X_i}(L)-v^j\clift{X_j}(\vlift{X_i}(L)).
\]
Now $\vlift{X_i}(\vlift{X_j}(L))$ are just the components of the
Hessian of $L$ expressed in terms of the anholonomic frame (recall
that vertical lifts commute):\ so these equations uniquely determine
the coefficients $\Gamma^i$, as required. We sum this discussion
up in the form of a proposition, which we regard as the fundamental
statement of regular Lagrangian dynamics.

\begin{prop}\label{Prop1}
Let $L$ be a regular Lagrangian on $TQ$.  There is a unique
second-order differential equation field $\Gamma$ such that
\[
\Gamma(\vlift{Z}(L))-\clift{Z}(L)=0
\]
for all vector fields $Z$ on $Q$. Moreover, $\Gamma$ may be determined
from the equations
\[
\Gamma(\vlift{X_i}(L))-\clift{X_i}(L)=0, \quad i=1,2,\ldots n,
\]
for any frame $\{X_i\}$ on $Q$ (which may be a coordinate frame or may be
anholonomic).
\end{prop}

The integral curves of $\Gamma$ are the solutions of the conventional
Euler-Lagrange equations of $L$, and we may therefore regard the first
displayed equation in the proposition as the most general form of the
Euler-Lagrange equations, and the second displayed equations as the
Euler-Lagrange equations relative to a frame.

Note that we do not assume that the Lagrangian is of mechanical type:\
the remarks above apply to any Lagrangian, provided it is regular.

These general forms of the Euler-Lagrange equations are very useful in
the discussion of theoretical questions, but in any specific
particular case it will eventually be necessary to introduce
coordinates.  One way of doing so, which partially takes cognizance of
an anholonomic frame, is to use some arbitrary coordinates $q^i$ on
$Q$, but the quasi-velocities $v^i$ of the frame for fibre
coordinates.  We next write the Euler-Lagrange equations for
the frame in terms of such coordinates.  To do so we need to express
$\clift{X_i}$ and $\vlift{X_i}$ in terms of the coordinate vector
fields, and this is most easily done by simply evaluating them on the
coordinates.  We have
\[
\clift{X_i}(q^j)=X_i(q^j)=X_i^j,\quad
\clift{X_i}(v^j)=-R^j_{ik}v^k,\quad
\vlift{X_i}(q^j)=0,\quad
\vlift{X_i}(v^j)=\delta^j_i,
\]
where $X_i=X_i^j\partial/\partial q^j$.  It follows that
\[
\clift{X_i}=X_i^j\vf{q^j}-R^j_{ik}v^k\vf{v^j},\quad
\vlift{X_i}=\vf{v^i}.
\]
The Euler-Lagrange equations become
\[
\Gamma\left(\fpd{L}{v^i}\right)-X_i^j\fpd{L}{q^j}+R^j_{ik}v^k\fpd{L}{v^j}=0,
\]
which are Hamel's equations in our notation.

\subsection{The method of the Cartan form}

The discussion above is limited to regular Lagrangians.  Later in the
paper we shall deal with a variational problem (the vakonomic problem)
for which the Lagrangian is very definitely not regular, and so the
methods developed so far will not apply.  To cope with this situation
we shall use the method of the Cartan forms, which are in fact globally defined forms. However, once again we need a
formulation in terms of a (local) anholonomic frame, which makes the
following account unusual in some respects.

We consider any Lagrangian $L$ on any $TQ$ equipped with the frame
lifted from an anholonomic frame $\{X_i\}$ on $Q$ whose dual is
$\{\theta^i\}$ and whose quasi-velocities are $v^i$.  The Cartan 1-form
of $L$ is $\vlift{X_i}(L)\theta^i$, the Cartan 2-form
\[
\omega=d\left(\vlift{X_i}(L)\theta^i\right).
\]
The energy of $L$ is $E=v^i\vlift{X_i}(L)-L$, and any vector field
$\Gamma$ satisfying
\[
\Gamma\hook\omega=-dE
\]
is an Euler-Lagrange field of $L$. A word of caution is required:\
once again we have not distinguished notationally between an object
on $Q$ and its pullback to $TQ$. Strictly speaking we should write
$\vlift{X_i}(L)\tau^*\theta^i$ for the Cartan 1-form. Bearing this in
mind it is easy to see that for any 1-form $\theta$ and vector field
$Z$ on $Q$,
\[
\theta(\clift{Z})=\theta(Z),\quad
\theta(\vlift{Z})=0,\quad
\clift{Z}\hook d\theta=Z\hook d\theta,\quad
\vlift{Z}\hook d\theta=0.
\]
In the case of the anholonomic frame and its dual, the third of these
leads to
\[
\clift{X_i}\hook d\theta^j=X_i\hook d\theta^j
=\lie{X_i}\theta^j-d(\theta^j(X_i))=-R_{ik}^j\theta^k.
\]
The equation $\Gamma\hook\omega=-dE$ is of course an equation between
1-forms on $TQ$. Evaluating it on $\vlift{X_i}$ and $\clift{X_i}$ in
turn leads to the pair of equations
\[
\vlift{X_i}(\vlift{X_j}(L))\left(\theta^j(\Gamma)-v^j\right)=0,\quad
\Gamma(\vlift{X_i}(L))-\clift{X_i}(L)=0.
\]
From the first of these, $\theta^j(\Gamma)=v^j$ when $L$ is regular,
which says that $\Gamma$ is a second-order differential equation
field; and then the second gives the Euler-Lagrange equations relative
to the frame.  But the equation $\Gamma\hook\omega=-dE$ stands even
when the Lagrangian is not regular; and in particular any
Euler-Lagrange field $\Gamma$ satisfies the equations
$\Gamma(\vlift{X_i}(L))-\clift{X_i}(L)=0$, though these equations may
not be enough to determine an Euler-Lagrange field uniquely.

\section{Nonholonomic dynamics}

We now consider a dynamical system subject to nonholonomic linear
constraints.  There are two (equivalent) ways of specifying such
constraints:\ as a distribution $\D$ on $Q$ (the constraint
distribution), such that at each $q\in Q$, $\D_q$ has the same
dimension $m$, and which is not integrable (in the sense of
Frobenius); or as a submanifold $\C$ of $TQ$ (the constraint
submanifold) which intersects each fibre in a linear subspace $\C_q$,
again of constant dimension $m$; of course $\C_q$ and $\D_q$ are the
same, just viewed from slightly different perspectives.

We can choose an anholonomic frame $\{X_i\}$ adapted to this situation
by taking its first $m$ members, say $\{X_\alpha\}$,
$\alpha=1,2,\ldots,m$, to span $\D$.  We write the remaining members
of the frame as $\{X_a\}$, $a=m+1,m+2,\ldots,n$.  Let $v^i$ be the
corresponding quasi-velocities:\ then $\C$ is the submanifold where
$v^a=0$.

We shall be interested in vector fields defined only on the constraint
submanifold $\C$.  A vector field $\Gamma$ on $\C$ (which for the
purpose of the following definition could be any submanifold of $TQ$)
will be said to be of second-order type if it satisfies the defining
condition for second-order differential equation fields,
$\tau_{*(q,u)}\Gamma=u$, for all $(q,u)\in\C$.  We shall furthermore be
interested in vector fields not only defined on the constraint
submanifold $\C$, but also tangent to it.  A vector field $\Gamma$ on
$\C$ will be tangent to $\C$ if and only if $\Gamma(v^a)=0$,
$a=m+1,m+2,\ldots,n$.  A vector field $\Gamma$ on $\C$ which is of
second-order type takes the form
\[
\Gamma=v^\alpha\clift{X_\alpha}+\Gamma^i\vlift{X_i}
\]
for some functions $\Gamma^i$ on $\C$. Now
\[
\clift{X_\alpha}(v^a)=-R^a_{\alpha i}v^i=-R^a_{\alpha\beta}v^\beta
\mbox{ on $\C$},
\]
so $v^\alpha\clift{X_\alpha}(v^a)=0$ since $R^a_{\alpha\beta}$ is
skew-symmetric in its lower indices.  But
$\vlift{X_i}(v^a)=\delta^a_i$, so for $\Gamma$ to be tangent to $\C$
we must have $\Gamma^a=0$.  That is to say, a vector field $\Gamma$ on
$\C$ which is of second-order type and which is tangent to $\C$ takes
the form
\[
\Gamma=v^\alpha\clift{X_\alpha}+\Gamma^\alpha\vlift{X_\alpha}
\]
with respect to the lifted anholonomic frame.

It is interesting to note that $\Gamma$, in the form just derived,
depends only on those vector fields of the frame which span $\D$. Let
us derive its representation with respect to another frame $\{Y_i\}$,
where again the $Y_\alpha$ span $\D$. Then
\[
Y_\alpha=A_\alpha^\beta X_\beta,\quad
Y_a=A_a^b X_b+A_a^\alpha X_\alpha,
\]
where the square matrices $(A_\alpha^\beta)$ and $(A_a^b)$, whose
entries are local functions on $Q$, are non-singular.  The
quasi-velocities $w^i$ corresponding to the new frame are given by
\[
w^a=\bar{A}^a_bv^b,\quad
w^\alpha=\bar{A}^\alpha_\beta(v^\beta-A^\beta_b\bar{A}^b_av^a),
\]
where the overbar indicates the inverse matrix.  Note that the level
sets $v^a=0$ and $w^a=0$ coincide, and that where $v^a=0$,
$w^\alpha=\bar{A}^\alpha_\beta v^\beta$. After a short calculation we
find that
\[
\Gamma=v^\alpha\clift{X_\alpha}+\Gamma^\alpha\vlift{X_\alpha}
=w^\alpha\clift{Y_\alpha}
+\bar{A}^\alpha_\beta(\Gamma^\beta
-\dot{A}^\beta_\gamma w^\gamma)\vlift{Y_\alpha}.
\]
This shows that indeed any vector field of second-order type on the
constraint submanifold $\C$ which is tangent to $\C$ can be expressed
entirely in terms of a local basis for $\D$, and incidentally gives
the transformation rule for the coefficients $\Gamma^\alpha$ under a
change of such a basis.

We now suppose given a Lagrangian $L$ on $TQ$.  We say that $L$ is
regular with respect to the constraints if for any local basis
$\{X_\alpha\}$ of $\D$, the symmetric $m\times m$ matrix whose entries
are $\vlift{X_\alpha}(\vlift{X_\beta}(L))$, functions on $\C$, is
nonsingular (this condition is easily seen to be independent of the
choice of basis).

\begin{prop}
Let $\C\subset TQ$ be the constraint submanifold for a system of
nonholonomic linear constraints, $\D$ the corresponding constraint
distribution, and $L$ a Lagrangian on $TQ$ which is regular with
respect to $\D$.  Then there is a unique vector field $\Gamma$ on $\C$
which is of second-order type, is tangent to $\C$, and is such that on
$\C$
\[
\Gamma(\vlift{Z}(L))-\clift{Z}(L)=0
\]
for all $Z\in\D$.  Moreover, $\Gamma$ may be determined from the
equations
\[
\Gamma(\vlift{X_\alpha}(L))-\clift{X_\alpha}(L)=0, \quad
\alpha=1,2,\ldots m,
\]
on $\C$, where $\{X_\alpha\}$ is any local basis for $\D$.
\end{prop}

\begin{proof}
For any $L$ and $\Gamma$, the map
$Z\mapsto\Gamma(\vlift{Z}(L))-\clift{Z}(L)$ is $\C^\infty(Q)$-linear,
just as before. Thus $\Gamma(\vlift{Z}(L))-\clift{Z}(L)=0$ if and
only if $\Gamma(\vlift{X_\alpha}(L))-\clift{X_\alpha}(L)=0$ for a
local basis $\{X_\alpha\}$. But $\Gamma$ must take the form
\[
\Gamma=v^\alpha\clift{X_\alpha}+\Gamma^\alpha\vlift{X_\alpha},
\]
so that
\[
\vlift{X_\alpha}(\vlift{X_\beta}(L))\Gamma^\beta =
\clift{X_\alpha}(L)-v^\beta\clift{X_\beta}(\vlift{X_\alpha}(L)),
\]
which determines $\Gamma^\alpha$ by the regularity assumption.
\end{proof}

We regard the content of this proposition as the fundamental statement
of regular nonholonomic dynamics; it is our version of the
Lagrange-d'Alembert principle.  The displayed equations in the
statement of the proposition are the fundamental equations of regular
nonholonomic dynamics.  We could express the main content of the
proposition in terms of the map $\varepsilon:\vectorfields{Q}\to
C^\infty(TQ)$ introduced earlier by saying that $\varepsilon$ takes
its values in $\D^0$, the annihilator of $\D$.  To be more precise:\
$\D^0$ is a linear subbundle of $T^*Q$, and $\varepsilon$, which in
general is a 1-form along the tangent bundle projection $\tau$,
according to the proposition is a section of $\tau_\C^*\D^0$ (where
$\tau_\C:\C\to Q$ is the restriction of $\tau$ to $\C$).

For the proposition we do not need a full anholonomic frame.  What
purpose might the remaining vectors $\{X_a\}$ of such a frame serve,
one might ask.  Here are two uses for them.

In many formulations of the equations of nonholonomic dynamics,
multipliers appear (see \cite{BMZ} for example).  Having determined
$\Gamma$ by the method of the proposition, we may then form the
expressions $\Gamma(\vlift{X_a}(L))-\clift{X_a}(L)$. These will not be
zero:\ let us set
\[
\Gamma(\vlift{X_a}(L))-\clift{X_a}(L)=\lambda_a.
\]
We have $X_i=X_i^j\partial/\partial q^j$ with $(X_i^j)$ nonsingular.
Then the coordinate version of the full set of equations reads
\[
\Gamma\left(\fpd{L}{u^i}\right)-\fpd{L}{q^i}=\bar{X}_i^a\lambda_a,
\]
where $\bar{X}_i^a$ are the appropriate entries in the matrix inverse
to $(X_i^j)$.  In fact $v^a=\bar{X}_i^au^i$, which explains the
significance of these coefficients:\ the constraint equations are just
$\bar{X}_i^au^i=0$.  We emphasise that the multipliers $\lambda_a$ are
determined once one has chosen a full anholonomic frame adapted to the
constraint distribution and found $\Gamma$ from the fundamental
equations.  In fact the multipliers
are the components of $\varepsilon$ with respect to the coframe dual
to the (full) frame; the notional components with subscript $\alpha$
vanish because $\varepsilon$ takes its values in $\D^0$.

The fundamental equations
$\Gamma(\vlift{X_\alpha}(L))-\clift{X_\alpha}(L)=0$ involve
differentiations of $L$ in directions transverse to $\C$.  Actually
$\vlift{X_\alpha}(v^a)=0$, so the restriction of $\vlift{X_\alpha}$ to
$\C$ is tangent to it.  But as we pointed out earlier,
$\clift{X_\alpha}(v^a)=-R^a_{\alpha\beta}v^\beta$ on $\C$.  Now
$R^a_{\alpha\beta}$ is the component of $[X_\alpha,X_\beta]$ along
$X_a$, and in the case of genuinely nonholonomic constraints some of
these components will be nonzero; so we must expect that
$\clift{X_\alpha}$ will not be tangent to $\C$.  Using the $X_a$ we
can split $\clift{X_\alpha}$ into a component tangent to $\C$, which
we denote by $\clift{\bar{X}_\alpha}$ (the notation is not
intended to imply that this is a complete lift, nor that an inverse
matrix is involved), and a component transverse to it.  Since
$\vlift{X_a}(v^b)=\delta^b_a$, we see that
\[
\clift{\bar{X}_\alpha}=\clift{X_\alpha}+R^a_{\alpha\beta}v^\beta\vlift{X_a}
\]
is tangent to $\C$.  Since $v^\alpha \clift{X_\alpha} =
v^\alpha \clift{{\bar X}_\alpha}$, the vector field $\Gamma$ is in
fact of the form
$\Gamma=v^\alpha\clift{{\bar X}_\alpha}+\Gamma^\alpha\vlift{X_\alpha}$.
We may therefore write the fundamental equations as
\[
\Gamma(\vlift{X_\alpha}(L))-\clift{\bar{X}_\alpha}(L)
=-R^a_{\alpha\beta}v^\beta\vlift{X_a}(L);
\]
now every term on the left-hand side depends only on the value of $L$
on $\C$.  Let us denote the restriction of $L$ to $\C$ by
$L_c$; this is often called the constrained Lagrangian. Then the
version of the fundamental equations above may be written
\[
\Gamma(\vlift{X_\alpha}(L_c))-\clift{\bar{X}_\alpha}(L_c)
=-R^a_{\alpha\beta}v^\beta\vlift{X_a}(L).
\]
The equations appear in a somewhat similar form to this in
\cite{FandB} (Equation~(1.6)).

Finally, we can easily write the fundamental equations in Hamel form:\
they are
\[
\Gamma\left(\fpd{L}{v^\alpha}\right)-X_\alpha^i\fpd{L}{q^i}
+R^i_{\alpha\beta}v^\beta\fpd{L}{v^i}=0
\]
(compare with \cite{BMZ} Theorem~3.2).

\section{Vakonomic systems}

The term `vakonomic mechanics' was introduced by Arnold and Kozlov in
\cite{Arnold,Kozlov} and stands for `mechanics of Variational
Axiomatic Kind'.  The theory was proposed as an alternative to the
approach to nonholonomic dynamics discussed in the previous section. A fundamental reference in this context is the book chapter \cite{Gersch} by Vershik and Gershkovich. Since more and more evidence seems to suggest that the vakonomic
equations do not give the correct equations of motion, some authors
have refrained from using the term vakonomic mechanics and prefer to
call these systems `variational nonholonomic systems' (see e.g.\
\cite{Bloch}).  We shall follow the majority in using the word
`vakonomic', but avoid controversy by talking of `vakonomic systems',
`vakonomic problems' and so on, but never mentioning `vakonomic
mechanics'.

One idea behind the
formulation of vakonomic systems is to think of the
multipliers as additional variables.  Recall that the multipliers are
in fact components of a 1-form $\varepsilon$ (along a certain
projection) which takes its values in $\D^0\subset T^*Q$, the
annihilator bundle of the constraint distribution $\D$.  We therefore
take $\D^0$ as state space for the vakonomic system.
The Lagrangian of the vakonomic system will therefore
be a function $\hat{L}$ on $T\D^0$.  It is constructed as follows.
First, we are given a Lagrangian $L$ on $TQ$; but $T\D^0$ projects
onto $TQ$, so we can pull $L$ back to $T\D^0$ (and as before we
denote the pulled-back function by the same symbol $L$).  Secondly,
every point $\mu$ of $\D^0$ over $q\in Q$ is a covector at $q$, and
so defines a linear function on $T_qQ$; we can therefore define a
function $M$ on $T\D^0$ by $M(q,\mu,u,\nu)=\mu(u)$ (note that $M$ is
independent of the second fibre component $\nu$).  Then
\[
\hat{L}=L-M.
\]

In keeping with the philosophy of the rest of this paper, we now
express $\hat{L}$ in terms of an anholonomic frame.  We introduce an
anholonomic frame $\{X_\alpha,X_a\}$ adapted to the constraint
distribution as before.  By doing so we effectively identify $\D^0$
with $Q\times\R^{n-m}$, or in other words we fix fibre coordinates
$\mu_a$ on $\D^0$, which are the components of $\mu$ with respect to
the coframe dual to the chosen frame. This  identification implies a local character of some of our results. For a more intrinsic formulation, see e.g.\ \cite{CDMM,Gersch}. The Lagrangian $\hat{L}$
is given as a function on $T(Q\times\R^{n-m})$ by
\[
\hat{L}=L-\mu_av^a
\]
where the $v^a$ are quasi-velocities, as before.

We assume that $L$ is regular.  But even so $\hat{L}$ fails decisively
to be regular, so we cannot obtain its Euler-Lagrange equations simply
by applying Proposition~\ref{Prop1}.  We shall instead use the method of the
Cartan form, as outlined in Section~2.3, to obtain them.

We may augment our frame $\{X_i\}$ on $Q$ to a frame on $Q\times
\R^{n-m}$ simply by adjoining the coordinate fields
$\partial/\partial\mu_a$.  The vector fields $X_i$, now interpreted as
vector fields on $Q\times \R^{n-m}$, act in the same way as before on
the coordinates of $Q$ and have the property that $X_i(\mu_a)=0$.
Since $\hat{L}$ does not depend on the velocity variables
corresponding to the $\mu_a$ its Cartan 1-form is
\[
\vlift{X_i}(\hat{L})\theta^i=\vlift{X_i}(L)\theta^i-\mu_a\theta^a.
\]
For the same reason, $\hat{E}=E$ (the energy of $L$). The equation
$\hat{\Gamma}\hook\hat{\omega}=-d\hat{E}$, which is an equation for a
vector field or fields $\hat{\Gamma}$ on $T(Q\times\R^{n-m})$, becomes
\[
\hat{\Gamma}\hook(\omega-d(\mu_a\theta^a))=-dE,
\]
where $\omega$ is the Cartan 2-form of $L$ (properly speaking, pulled
back to $T(Q\times\R^{n-m})$).  By evaluating the equation on
$\vlift{X_i}$ we obtain
\[
\vlift{X_i}(\vlift{X_j}(L))\left(\theta^j(\hat{\Gamma})-v^j\right)=0,
\]
whence $\theta^j(\hat{\Gamma})=v^j$ since we assume that $L$ is
regular.  The coefficient of the term in $d\mu_a$ must vanish, whence
$\theta^a(\hat{\Gamma})=0=v^a$:\ that is to say, the equation
$\hat{\Gamma}\hook\hat{\omega}=-d\hat{E}$ has no solution except where
$v^a=0$, that is, except on $\C\times T\R^{n-m}$.  Finally, evaluating
the equation on $\clift{X_\alpha}$ and $\clift{X_a}$ successively gives
\begin{align*}
\hat{\Gamma}(\vlift{X_\alpha}(L))-\clift{X_\alpha}(L)
&=\mu_aR^a_{\alpha\beta}v^\beta,\\
\hat{\Gamma}(\vlift{X_a}(L))-\clift{X_a}(L)
&=\mu_bR^b_{a\alpha}v^\alpha+\hat{\Gamma}(\mu_a).
\end{align*}
We shall sometimes combine these two sets of equations into one by
writing
\[
\hat{\Gamma}(\vlift{X_i}(L))-\clift{X_i}(L)
=\mu_aR^a_{i\alpha}v^\alpha+\hat{\Gamma}(\mu_a)\delta^a_i.
\]
Now $\hat{\Gamma}$ is a vector field on $\C\times T\R^{n-m}\subset
T(Q\times\R^{n-m})$. It is natural to decompose it according to the
product structure, say $\hat{\Gamma}=\Gamma_{\C}+\Gamma_\mu$. The
equation displayed above then becomes
\[
\Gamma_{\C}(\vlift{X_i}(L))-\clift{X_i}(L)
=\mu_aR^a_{i\alpha}v^\alpha+\Gamma_\mu(\mu_a)\delta^a_i.
\]
We may construct Euler-Lagrange fields
$\hat{\Gamma}=\Gamma_{\C}+\Gamma_\mu$ for the vakonomic problem as
follows:\ $\Gamma_\mu$ is completely undetermined; but a choice of
$\Gamma_\mu$ having been made, $\Gamma_{\C}$ will in favourable
circumstances be determined by the equation above.

Since the velocity variables corresponding to the $\mu_a$ do not
appear in the equations, and in fact play no part at all in the
theory, we propose to ignore them; that is to say, we shall replace
$T(Q\times\R^{n-m})$ by $TQ\times\R^{n-m}$ and $\C\times T\R^{n-m}$ by
$\C\times\R^{n-m}$.  The frame $\{X_i,\partial/\partial\mu_a\}$ can
be lifted to a frame
$\{\clift{X}_i,\vlift{X}_i,\partial/\partial\mu_a\}$ (and in
principle  also
$\partial/\partial \nu_a$, but we ignore these in view of what was
said above).  In this frame, the vector fields $\clift{X}_i$, for
example, should be interpreted as vector fields on $TQ\times
\R^{n-m}$:\ again, they act in the same way as before on the
coordinate functions $x^i,v^i$ and have the property that
$\clift{X_i}(\mu_a)=0$.

Furthermore, the fact that the $\mu_a$ are variables is really rather
an embarrassment; we would prefer to think of them as functions on
$\C$, or in other words to take a section $\phi$ of the projection
$\C\times\R^{n-m}\to\C$ and restrict our attention to its image
$\im(\phi)$.

Any Euler-Lagrange field $\hat{\Gamma}$, restricted to $\im(\phi)$,
will take the form $\hat{\Gamma}=\Gamma_{\C}+\Gamma_\mu$ where since
$\Gamma_{\C}$ is of second-order type
\[
\Gamma_{\C}=v^\alpha\clift{X_\alpha}
+\Gamma^\alpha\vlift{X_\alpha}+\Gamma^a\vlift{X_a},\quad\mbox{and}\quad
\Gamma_\mu=A_a\vf{\mu_a};
\]
all coefficients are functions on $\C$.  Notice that though
$\Gamma_{\C}$ is here defined on $\C$, there is no reason
to believe that it is tangent to it; hence the inclusion of the term
in $\vlift{X_a}$.

 We may conclude:
\begin{prop}
If a vakonomic Euler-Lagrange field $\hat\Gamma$ is decomposed as above, the restriction of the vakonomic Euler-Lagrange equations to
$\im(\phi)$ may be written
\begin{align*}
\Gamma_\C(\vlift{X_\alpha}(L))-\clift{X_\alpha}(L)&=
\phi_aR^a_{\alpha\beta}v^\beta,\\
\Gamma_\C(\vlift{X_a}(L))-\clift{X_a}(L)&=
\phi_bR^b_{a\alpha\beta}v^\alpha+A_a.
\end{align*}
\end{prop}
We shall always use the vakonomic equations in this form.

The first point to note is that provided $L$ is regular, so that the
Hessian of $L$ is non-singular, for given $\phi$ and $A_a$ these
equations determine $\Gamma_\C$ as a vector field on $\C\subset TQ$.
We next show that when further regularity conditions are satisfied we
can choose $A_a$ such that $\Gamma_\C$ is not just a vector field on
$\C$ but a vector field tangent to $\C$:\ that is, we can choose $A_a$
such that $\Gamma^a=0$.

For convenience we denote the components of the Hessian of $L$ with
respect to the frame $\{X_i\}$ by $g_{ij}$:
\[
\vlift{X_\alpha}(\vlift{X_\beta}(L))=g_{\alpha\beta},\quad
\vlift{X_\alpha}(\vlift{X_a}(L))=g_{\alpha a},\quad
\vlift{X_a}(\vlift{X_b}(L))=g_{ab}.
\]
The $g_{ij}$ are symmetric in their indices.  In discussing the
Lagrange-d'Alembert principle we imposed a regularity condition on
$L$, namely that the submatrix $(g_{\alpha\beta})$ of its Hessian must
be nonsingular on $\C$:\ we said that $L$ is then regular with respect
to $\D$.  Assuming that
$L$ is indeed regular with respect to $\D$, let $(g^{\alpha\beta})$ be
the matrix inverse to $(g_{\alpha\beta})$.  In the following argument
we must assume that $(g_{ab}-g^{\alpha\beta}g_{a\alpha}g_{b\beta})$ is
nonsingular on $\C$.  This matrix is in fact the restriction of the
Hessian $(g_{ij}(q,u))$ to the subspace of $T_qQ$ which is orthogonal to
$\D_q$ with respect to it, so when this matrix is nonsingular we say that
$L$ is regular with respect to $\D^\perp$.  When $(g_{ij})$ is
positive definite, $(g_{ab}-g^{\alpha\beta}g_{a\alpha}g_{b\beta})$
will automatically be nonsingular, as will $(g_{\alpha\beta})$ be;
indeed, both will be positive definite.  So when $(g_{ij})$ is
positive definite $L$ will automatically be regular with respect to
both $\D$ and $\D^\perp$, but for Hessians of other signatures we need
to make the assumptions explicit.

The equations for $\Gamma_{\C}$ can be written
\begin{align*}
&g_{\alpha\beta}\Gamma^\beta+g_{\alpha b}\Gamma^b=Y_\alpha,\\
&g_{a\beta}\Gamma^\beta+g_{ab}\Gamma^b=Y_a+A_a,
\end{align*}
where $Y_\alpha$ and $Y_a$ are known expressions in $L$, the $X_i$,
$\phi$, the $R^a_{ij}$, etc.  Using the assumption that $g_{\alpha\beta}$ is
nonsingular we can eliminate $\Gamma^\beta$ between these equations,
leaving
\[
(g_{ab}-g^{\alpha\beta}g_{a\alpha}g_{b\beta})\Gamma^b=
A_a+Y_a-g_{a\alpha}g^{\alpha\beta}Y_\beta.
\]
So when $(g_{ab}-g^{\alpha\beta}g_{a\alpha}g_{b\beta})$ is
nonsingular, by taking $A_a=g_{a\alpha}g^{\alpha\beta}Y_\beta-Y_a$ we
ensure that $\Gamma^a=0$ as required.

The resultant vector field $\Gamma_{\C}$, which is the projection onto
$\C$ of the restriction of $\hat{\Gamma}$ to $\im(\phi)$, is
determined by the equations
\[
\Gamma_{\C}(\vlift{X_\alpha}(L))-\clift{X_\alpha}(L)
=\phi_aR^a_{\alpha\beta}v^\beta.
\]
Of course, it will also satisfy
\[
\Gamma_{\C}(\vlift{X_a}(L))-\clift{X_a}(L)
=\phi_bR^b_{a\alpha}v^\alpha+A_a
\]
(with $A_a=g_{a\alpha}g^{\alpha\beta}Y_\beta-Y_a$); that is to say,
\[
\hat{\Gamma}=\Gamma_\C+\Gamma_\mu=v^\alpha\clift{X_\alpha}
+\Gamma^\alpha\vlift{X_\alpha}+A_a\vf{\mu_a}
\]
will satisfy the vakonomic equations on $\im(\phi)$.

Let us set $\Gamma_{\C}(\vlift{X_a}(L))-\clift{X_a}(L)=\Lambda_a$.
Then $A_a=\Lambda_a-\phi_bR^b_{a\alpha}v^\alpha$.

Now $\hat{\Gamma}$, so determined, is a vector field on $\im(\phi)$.
A natural question to ask is whether it is tangent to $\im(\phi)$,
that is, whether $\hat{\Gamma}(\mu_a-\phi_a)=0$ for $\mu_a=\phi_a$.
But $\hat{\Gamma}(\mu_a-\phi_a)=A_a-\Gamma_{\C}(\phi_a)$.  The
condition for tangency is therefore
\[
\Gamma_{\C}(\phi_a)+\phi_bR^b_{a\alpha}v^\alpha=\Lambda_a.
\]

We can summarise this discussion as follows.

\begin{prop}\label{GammaC}
Assume that $L$ is regular with respect to both $\D$ and $\D^\perp$.
Let $\phi:\mu_a=\phi_a$ be a section of $\C\times\R^{n-m}\to\C$.
There is a unique vector field $\hat{\Gamma}$ on the image of $\phi$
which is a solution of the vakonomic problem there, and
is such that its projection onto $\C$, $\Gamma_{\C}$, is tangent to
$\C$.  The vector field $\hat{\Gamma}$ so determined is tangent to the
image of $\phi$ if and only if
$\Gamma_{\C}(\phi_a)+\phi_bR^b_{a\alpha}v^\alpha=\Lambda_a$.
\end{prop}

One might very well query the importance of the requirement of
tangency to $\im(\phi)$:\ but consider the following.

Suppose that the $\phi_a$ can be continued off $\C$, that is, suppose
that there are functions $\Phi_a$ defined in a neighbourhood of $\C$
in $TQ$ such that $\phi_a=\Phi_a|_{\C}$.  This can always be done
locally:\ since $\C$ is given by $v^a=0$ it is enough to make $\Phi_a$
independent of the $v^a$.  Suppose that the $A_a$ have been chosen so
that $\Gamma_{\C}$ is tangent to $\C$; then
$\Gamma_{\C}(\Phi_a)=\Gamma_{\C}(\phi_a)$.  Consider the Lagrangian
$\tilde{L}=L-\Phi_av^a$, which is, note, a function on some
neighbourhood of $\C$ in $TQ$.  We have
\begin{align*}
\vlift{X_i}(\tilde{L})&=\vlift{X_i}(L)-\vlift{X_i}(\Phi_a)v^a-\Phi_a\delta^a_i,\\
\clift{X_i}(\tilde{L})&=\clift{X_i}(L)-\clift{X_i}(\Phi_a)v^a+\Phi_aR^a_{ij}v^j,
\end{align*}
whence on $\C$
\[
\Gamma_{\C}(\vlift{X_\alpha}(\tilde{L}))-\clift{X_\alpha}(\tilde{L})=
\Gamma_{\C}(\vlift{X_\alpha}(L))-\clift{X_\alpha}(L)
-\phi_aR^a_{\alpha\beta}v^\beta=0,
\]
while
\begin{align*}
\Gamma_{\C}(\vlift{X_a}(\tilde{L}))-\clift{X_a}(\tilde{L})&=
\Gamma_{\C}(\vlift{X_a}(L))-\clift{X_a}(L)
-\Gamma_{\C}(\phi_a)-\phi_bR^b_{a\alpha}v^\alpha\\
&=\Lambda_a-\Gamma_{\C}(\phi_a)-\phi_bR^b_{a\alpha}v^\alpha.
\end{align*}
Thus the tangency condition is the necessary and sufficient condition
for $\Gamma_{\C}$ to satisfy the Euler-Lagrange equations of
$\tilde{L}$ on $\C$ (and this for any extension of the $\phi_a$ off
$\C$).  That is to say:
\begin{thm}\label{Ltilde}
If there is a section $\phi$ such that
$\hat{\Gamma}=\Gamma_{\C}+\Gamma_\mu$ is tangent to the image of
$\phi$, where $\Gamma_\mu$ is chosen such that $\Gamma_{\C}$ is
tangent to $\C$, then $\Gamma_{\C}$ is the restriction to $\C$ of an
Euler-Lagrange field of $\tilde{L}$ for any extension of the $\phi_a$
off $\C$.  Conversely, let $\tilde{\Gamma}$ be an Euler-Lagrange field
of $\tilde{L}=L-\Phi_av^a$; suppose that $\tilde{\Gamma}$ is tangent
to $\C$.  Let $\phi_a=\Phi_a|_{\C}$ and consider the section
$\phi:\mu_a=\phi_a$.  Then $\phi_*\tilde{\Gamma}^0$, where
$\tilde{\Gamma}^0$ is the restriction of $\tilde{\Gamma}$ to $\C$, is
a solution of the vakonomic problem on the image of $\phi$.
\end{thm}

\section{The consistency problem}

We now turn to the vexed question of whether the nonholonomic and
vakonomic problems can ever be in any sense consistent.  As we have
shown, the vakonomic problem leads to a whole class of dynamical
vector fields $\hat{\Gamma}$ defined on $\C\times\R^{n-m}$ but not
necessarily tangent to it, the Lagrange-d'Alembert principle to a
single vector field $\Gamma$ defined on $\C$ and tangent to it.  The
question of how one decides whether there can be any coincidence
between such vector fields is therefore somewhat puzzling.  We propose
the following answer, under the assumption that $L$ is regular with
respect to both $\D$ and $\D^\perp$:\ we say that the vakonomic
and the nonholonomic problems are {\em weakly consistent\/} if there
is a section $\phi$ of the projection $\C\times \R^{n-m}\to\C$ such
that the corresponding vector field $\Gamma_\C$ defined in
Proposition~\ref{GammaC} is the nonholonomic dynamical vector field
$\Gamma$.  (We shall propose a stronger criterion for consistency
shortly.)

\begin{thm}
The following three statements are equivalent.
\begin{itemize}
\item[(i)] The variational nonholonomic
problem and the nonholonomic dynamics are weakly consistent.
\item[(ii)] There is a section $\phi$ of the projection $\C\times
\R^{n-m}\to\C$ such that $\phi_aR^a_{\alpha\beta}v^\beta=0$.
\item[(iii)] There is a section $\phi$ of the projection $\C\times
\R^{n-m}\to\C$ such that the vector field $\bar{\Gamma}$ on
$\im(\phi)$ given by
\[
\bar{\Gamma}=\Gamma+(\lambda_a-\phi_bR^b_{a\alpha}v^\alpha)\vf{\mu_a},
\]
where the $\lambda_a$ are the multipliers for $\Gamma$, satisfies the
vakonomic equations on $\im(\phi)$.
\end{itemize}
\label{Thm2}
\end{thm}

\begin{proof}
As we pointed out earlier, $\Gamma_\C$ is determined by the equation
\[
\Gamma_{\C}(\vlift{X_\alpha}(L))-\clift{X_\alpha}(L)
=\phi_aR^a_{\alpha\beta}v^\beta
\]
on $\C$. Thus $\Gamma_\C$ will coincide with $\Gamma$ if and only if
$\phi_aR^a_{\alpha\beta}v^\beta=0$.

On the other hand $\bar{\Gamma}$ satisfies
\begin{align*}
\bar{\Gamma}(\vlift{X_\alpha}(L))-\clift{X_\alpha}(L)&=0,\\
\bar{\Gamma}(\vlift{X_a}(L))-\clift{X_a}(L)
&=\phi_bR^b_{a\alpha}v^\alpha+\bar{\Gamma}(\mu_a),
\end{align*}
on $\im(\phi)$. Comparing with the vakonomic equations on $\im(\phi)$,
namely
\begin{align*}
\hat{\Gamma}(\vlift{X_\alpha}(L))-\clift{X_\alpha}(L)&=
\phi_aR^a_{\alpha\beta}v^\beta,\\
\hat{\Gamma}(\vlift{X_a}(L))-\clift{X_a}(L)&=
\phi_bR^b_{a\alpha}v^\alpha+\hat{\Gamma}(\mu_a),
\end{align*}
we see that $\bar{\Gamma}$ satisfies the vakonomic equations if and
only if $\phi_aR^a_{\alpha\beta}v^\beta=0$.
\end{proof}

Theorem~\ref{Thm2} shows that there is a process leading from the
nonholonomic dynamics $\Gamma$ to a vakonomic field $\bar{\Gamma}$ and
that therefore weak consistency works in both directions.  In fact the
construction leading to $\bar\Gamma$ can be applied to any vector
field of second-order type tangent to $\C$, with the obvious choice of
`multipliers' $\lambda_a$; if it is applied to $\Gamma_\C$ we get back
$\hat{\Gamma}$.  Conversely, if we apply Proposition~\ref{GammaC} to
$\bar{\Gamma}$ we obtain $\Gamma$.

The condition in item $(ii)$ of the theorem is, of course, not new.
It can also be found in other texts such as e.g.\ \cite{CDMM,FandB}
and it dates back to at least the paper \cite{Rumiantsev} of
Rumiantsev.  It can, however, always be
satisfied:\ one merely has to take $\phi_a=0$, that is, choose the
zero section.  The corresponding vakonomic field $\bar{\Gamma}$ then
has a particularly attractive form:
\[
\bar{\Gamma}=\Gamma+\lambda_a\vf{\mu_a}.
\]
The fact that weak consistency always holds explains why we regard it
as weak, and why we propose the following stronger version.  We
say that the nonholonomic and vakonomic problems are {\em strongly
consistent\/} if the problems are weakly consistent on $\im(\phi)$ for
some section $\phi$, and if in addition the vakonomic field
$\bar{\Gamma}$ is tangent to $\im(\phi)$, or in other words
$\phi_*\Gamma$ is vakonomic.

\begin{cor} \label{cor0}
The necessary and sufficient condition for the nonholonomic and
vakonomic problems to be strongly consistent is the existence of a
section $\phi$ of the projection $\C\times
\R^{n-m}\to\C$ such that $\phi_aR^a_{\alpha\beta}v^\beta=0$
and $\Gamma(\phi_a)+\phi_bR^b_{a\alpha}v^\alpha=\lambda_a$.
\end{cor}

\begin{cor}
If the nonholonomic and vakonomic problems are strongly consistent,
$\Gamma$ is the restriction to $\C$ of an Euler-Lagrange field of a
Lagrangian $\tilde{L}$.\label{cor1}
\end{cor}

\begin{proof}
Apply Theorem~\ref{Ltilde}.
\end{proof}

It seems to us that Corollary 2 provides what is probably the most
interesting consequence of our analysis of the consistency of
nonholonomic and vakonomic problems, namely that strong consistency is
a sufficient condition for the nonholonomic dynamics to be the
restriction of an Euler-Lagrange field.

The introduction of $\tilde{L}$ is the process called by Fernandez and
Bloch \cite{FandB}, in the context of Abelian Chaplygin systems,
`the elimination of the multipliers'.  To obtain their result these
authors assume what they call `conditional variationality', which in
that context corresponds to our strong consistency.
Corollary~\ref{cor1} above is a substantial generalization of their
result.  It is worth emphasising that Theorem~\ref{Ltilde} itself,
which is the major ingredient of the corollary, is quite independent
of the question of consistency and is concerned just with the
vakonomic problem.

Strong consistency has the following consequences so far as individual
motions are concerned.  Let $\gamma$ be any individual integral curve
of $\Gamma$.  Let $\psi^0_a$ be such that
\[
\psi^0_a(R^a_{\alpha\beta}v^\beta)(\gamma(0))=0.
\]
The equations
\[
\dot{\psi}_a(t)+\psi_b(t)R^b_{a\alpha}v^\alpha=\lambda_a,\quad
\psi_a(0)=\psi^0_a,
\]
where all of the variables other than $\psi_a$ are evaluated at
$\gamma(t)$, have a unique solution.  Then $(\gamma(t),\psi(t))$ is a
solution curve of the vakonomic problem if and only if
$\psi_a(t)R^a_{\alpha\beta}v^\beta=0$ for all $t$.  If strong
consistency holds, with respect to a section $\mu_a=\phi_a$, then
$\psi_a(t)=\phi_a(\gamma(t))$ satisfies this condition for every
$\gamma$.  On the other hand, it would be possible to formulate a
partial version of consistency, in which this condition holds along
some, but not necessarily all, integral curves of $\Gamma$. A
definition along these lines is to be found in \cite{FandB}. It seems likely that the conditions for partial consistency case are related to the so-called second-order constraints in the algorithm proposed in \cite{CDMM}.

\section{Chaplygin systems}

We now discuss the results of the previous section as they apply to
(non-Abelian) Chaplygin systems, which we define immediately below.

Assume that a Lie group $G$ acts in a free and proper way on the
configuration manifold $Q$.  Then $Q\to Q/G$ is a principal fibre
bundle.  A Chaplygin system \cite{Koiller}, sometimes referred to
as a generalized Chaplygin system \cite{Frans} or as a nonholomic
system of `purely kinematic' or `principal' type \cite{Bloch,BKMM}, is
a nonholonomic system where the Lagrangian is invariant under the
induced action of $G$ on $TQ$, and moreover the constraint distribution
$\D$ is the horizontal distribution of a principal connection on the
bundle. For more details about principal bundles and principal
connections, see e.g.\ \cite{KN}.

The most natural choice for a frame for a Chaplygin system is one
where the vector fields $X_a=\tilde{E}_a$ are the fundamental vector
fields of the action of the group $G$ and the vector fields $X_\alpha$
span $\D$ and in addition are invariant under $G$.  Thus
$[\tilde{E}_a,\tilde{E}_b]=-C_{ab}^c\tilde{E}_c$ where the
coefficients are the structure constants of the Lie algebra $\g$ of
$G$, and $[\tilde{E}_a,X_\alpha]=0$.  We have ${R}_{a\alpha}^i=0$ and
${R}_{ab}^i=-\delta^i_cC^c_{ab}$.  The formula for the action of a
complete lift on a quasi-velocity gives
\[
\clift{\tilde{E}_a}(v^b)=-{R}_{ai}^bv^i=C_{ac}^bv^c.
\]
In particular, $\clift{\tilde{E}_a}$ is tangent to $\C$:\ the group
action leaves $\C$ invariant.

The nonholonomic dynamical vector field $\Gamma$ obtained from the
Lagrange-d'Alembert principle is determined by the fundamental
equations
\[
\Gamma(\vlift{X_\alpha}(L))-\clift{X_\alpha}(L)=0.
\]
These may be written in the alternative form
\[
\Gamma(\vlift{X_\alpha}(L_c))-\clift{\bar{X}_\alpha}(L_c)
=-{R}^a_{\alpha\beta}v^\beta\vlift{\tilde{E}_a}(L)=
-{R}^a_{\alpha\beta}v^\beta p_a
\]
where $p_a=\vlift{\tilde{E}_a}(L)$ is just the $a$th component of
momentum for the free Lagrangian $L$ corresponding to the action of
$G$ as a symmetry group of $L$.  The component
$\clift{\bar{X}_\alpha}$ of $\clift{X_\alpha}$ tangent to $\C$ is
given in this case by
\[
\clift{\bar{X}_\alpha}=\clift{X_\alpha}+
{R}^a_{\alpha\beta}v^\beta\vlift{\tilde{E}_a}.
\]
Note that
$\clift{\bar{X}_\alpha}(L)=\clift{X_\alpha}(L)+
{R}^a_{\alpha\beta}v^\beta p_a$.

Since $\clift{\tilde{E}_a}(L)=0$, the multiplier equation is
\[
\Gamma(p_a)=\lambda_a;
\]
the momentum is not conserved by the nonholonomic dynamics, therefore,
but it is related to the multipliers in a simple fashion.
Due to the invariance of $L$, we have
$\clift{\tilde{E}_a}(p_b) =
[\clift{\tilde{E}_a},\vlift{\tilde{E}_b}](L) =
-C_{ab}^c\vlift{\tilde{E}_c}(L)$, or
\[
\clift{\tilde{E}_a}(p_b)+C_{ab}^cp_c=0,
\]
which says in fact that the momentum, regarded as a $\g^*$-valued
function on $TQ$, transforms according to the coadjoint action of $G$.
By acting with $\clift{\tilde{E}_a}$ on the fundamental equations one
finds that $[\clift{\tilde{E}_a},\Gamma](\vlift{X_\alpha}(L))=0$.  But
as we pointed out in Section~2, $[\clift{\tilde{E}_a},\Gamma]$ (the
bracket of a complete lift and a vector field of second-order type) is
vertical.  It is also tangent to $\C$.  Provided that $L$ is regular
with respect to $\D$ it follows that $[\clift{\tilde{E}_a},\Gamma]=0$,
which is to say that the dynamical vector field $\Gamma$ is invariant
under the action of $G$ on $\C$.
It follows immediately that
\[
\clift{\tilde{E}_a}(\lambda_b)+C_{ab}^c\lambda_c=0,
\]
which says that the multipliers should also be seen as
components of a $\g^*$-valued function transforming under the
coadjoint action.  In fact for $q\in Q$ the map $\g\to T_qQ$ by
$\xi\mapsto\tilde{\xi}_q$ identifies $\g$ with a subspace of $T_qQ$
complementary to $\D_q$ (namely the subspace tangent to the fibre).
Under this identification $\D^0_q$, the annihilator subsapce of $\D_q$
in $T^*_qQ$, can be identified with $\g^*$.

Consider now the vakonomic problem.  The Euler-Lagrange equations for
a vakonomic field $\hat{\Gamma}$ on $\C\times \R^{n-m}\subset
TQ\times\R^{n-m}$, namely
\[
\hat{\Gamma}(\vlift{X_i}(L))-\clift{X_i}(L)
=\mu_a{R}^a_{i\alpha}v^\alpha+\hat{\Gamma}(\mu_a)\delta^a_i,
\]
in this case become
\begin{align*}
\hat{\Gamma}(\vlift{X_\alpha}(L))-\clift{X_\alpha}(L)
&=\mu_a{R}^a_{\alpha\beta}v^\beta,\\
\hat{\Gamma}(\vlift{\tilde{E}_a}(L))&=\hat{\Gamma}(\mu_a).
\end{align*}
The latter just says that $p_a-\mu_a$, a function on
$TQ\times\R^{n-m}$, is a constant of motion for every vakonomic
dynamical vector field $\hat{\Gamma}$, or in other words that every
$\hat{\Gamma}$ is tangent to the level sets of $p_a-\mu_a$.  Since
$\hat{\Gamma}$ is defined only for $v^a=0$, we must in fact restrict
these functions to $\C\times\R^{n-m}$, and make the more comprehensive
statement that every $\hat{\Gamma}$ is tangent to the level sets of
$p_a-\mu_a$ in $\C\times\R^{n-m}$.

We observed in the previous section that it is preferable to restrict
the Euler-Lagrange equations of the vakonomic problem to the image of
some section $\phi$ of $\C\times\R^{n-m}\to\C$.  We have an obvious
choice of section in the present case, namely $\mu_a=p_a$, the zero
level set of $p_a-\mu_a$.  From the previous paragraph we see that
every vakonomic field $\hat{\Gamma}$ is tangent to $\im(\phi)$.  It
follows that with this choice of section, there is no difference
between weak and strong consistency for a Chaplygin system.  In
fact the second condition of Corollary~\ref{cor0} for strong
consistency, namely
$\Gamma(\phi_a)+\phi_bR^b_{a\alpha}v^\alpha=\lambda_a$, reduces to
$\Gamma(\phi_a)=\lambda_a$ with our choice of frame, and is satisfied
automatically when $\phi_a=p_a$.  Indeed, it is satisfied for
$\phi_a=p_a+k_a$ if $\Gamma(k_a)=0$, that is, if $k_a$ is a constant
of motion for the nonholonomic dynamics.  We work with the section
$\phi_a=p_a$ here; in the next section, however, we shall show that for
the two-wheeled carriage example the consistency conditions can
sometimes be satisfied with non-zero $k_a$.

We now consider the situation in the light of Theorem~\ref{Ltilde}.
This concerns the Euler-Lagrange equations of the Lagrangian
$\tilde{L}=L-\Phi_av^a$ where $\Phi_a$ is any extension of $\phi_a$
off $\C$.  In the present case (with $\phi_a=p_a$) there is no
difficulty in extending $\phi_a$ off $\C$:\ $\phi_a$ is just the
restriction to $\C$ of $p_a$, a well-defined function on $TQ$.  We can
therefore take $\tilde{L}=L-p_av^a$.  We shall discuss the dynamics of
this Lagrangian in the next few paragraphs.

We need to impose another regularity condition on $L$. The components
of the Hessian are denoted as follows:
\[
\vlift{X_\alpha}(\vlift{X_\beta}(L))=g_{\alpha\beta},\quad
\vlift{X_\alpha}(\vlift{\tilde{E}_a}(L))=g_{\alpha a},\quad
\vlift{\tilde{E}_a}(\vlift{\tilde{E}_b}(L))=g_{ab}.
\]
We now require that the submatrix $(g_{ab})$ of the Hessian is
nonsingular.  When this holds we say that $L$ is regular with respect
to $\g$.

We consider the Lagrangian $\tilde{L}=L-p_av^a$ on $TQ$.

\begin{prop}
If the Lagrangian $L$ of a Chaplygin system is regular with respect to $\D$ and to $\g$, the function $\tilde{L}=L-p_av^a$ is regular in some neighbourhood of $\C$ and invariant under the action of $G$.  Moreover, the corresponding Euler-Lagrange field $\tilde\Gamma$ is tangent to $\C$.
\end{prop}

\begin{proof}
A short
calculation leads to the following expressions for the components of
the Hessian of $\tilde{L}$, where $E$ stands for the vector field
$v^a\vlift{\tilde{E}_a}$ (so that $\tilde{L}=L-E(L)$):
\begin{align*}
&\vlift{X_\alpha}(\vlift{X_\beta}(\tilde{L}))=g_{\alpha\beta}
-E(g_{\alpha\beta}),\\
&\vlift{X_\alpha}(\vlift{\tilde{E}_a}(\tilde{L}))=
-E(g_{\alpha a}),\\
&\vlift{\tilde{E}_a}(\vlift{\tilde{E}_b}(\tilde{L}))=
-g_{ab}-E(g_{ab}).
\end{align*}
Note that $E=0$ on $\C$, from which it follows that if $L$ is regular
with respect to $\D$ and to $\g$ then $\tilde{L}$ is regular on $\C$,
and therefore regular in some neighbourhood of $\C$ at least.  In such
a neighbourhood of $\C$ there is a unique second-order differential
equation field $\tilde{\Gamma}$ which satisfies
$\tilde{\Gamma}(\vlift{X_i}(\tilde{L}))-\clift{X_i}(\tilde{L})=0$, the
Euler-Lagrange equations of $\tilde{L}$.

Now $\tilde{L}$ (as well as $L$) is invariant under the action of $G$
on $TQ$ :\ since $\clift{\tilde{E}_a}(p_b)+C_{ab}^cp_c=0$,
\[
\clift{\tilde{E}_a}(\tilde{L})=
\clift{\tilde{E}_a}(L-p_bv^b)=
-(\clift{\tilde{E}_a}(p_b)+p_cC^c_{ab})v^b
=0.
\]
The momentum ${\tilde p}_a = \vlift{\tilde{E}_a}(\tilde{L})$ will
therefore be conserved by the dynamical field $\tilde\Gamma$, or in
other words $\tilde{\Gamma}$ will be tangent to the level sets of
${\tilde p}_a$.  But
\[
{\tilde p}_a=
\vlift{\tilde{E}_a}(L-p_av^a)=
p_a-g_{ab}v^b-p_a=-g_{ab}v^b.
\]
Thus if $L$ is regular with respect to $\g$, the zero level of
${\tilde p}_a$ is precisely $\C$. We can conclude therefore that the Euler-Lagrange field
$\tilde{\Gamma}$ of $\tilde{L}$ is tangent to $\C$.
\end{proof}

 This proposition extends a
result given in \cite{FandB} for the case of an Abelian Chaplygin
system to Chaplygin systems in general.

\begin{cor} Under the assumptions of the previous proposition, the vector field $\phi_*\tilde{\Gamma}^0$,  where $\tilde{\Gamma}^0$ is the restriction of $\tilde{\Gamma}$ to
$\C$, is a solution of the vakonomic problem on the image of $\phi$.
\end{cor}
\begin{proof}
This follows from Theorem~\ref{Ltilde}.
\end{proof}

 The equations determining $\tilde{\Gamma}^0$ are now to be compared with those
for the nonholonomic dynamics.

\begin{prop}
Under the assumptions of the previous proposition, the restriction $\tilde{\Gamma}^0$ to $\C$ of the Euler-Lagrange field of  $\tilde L$ equals the nonholonomic field $\Gamma$ if and only if
\[
{R}^a_{\alpha\beta}v^\beta
p_a=0,
\]
or, equivalently, if and only if the conditions for weak (and thus strong) consistency are satisfied.
\end{prop}

\begin{proof}
The equations determining $\tilde{\Gamma}^0$ are of course the
restrictions to $\C$ of the Euler-Lagrange equations
$\tilde{\Gamma}(\vlift{X_i}(\tilde{L}))-\clift{X_i}(\tilde{L})=0$.
Those with $i=a$ just say that $\tilde{\Gamma}^0$ is tangent to $\C$.
The others may be written
\[
\tilde{\Gamma}^0(\vlift{X_\alpha}(\tilde{L}_c))
-\clift{\bar{X}_\alpha}(\tilde{L}_c)=
-{R}^a_{\alpha\beta}v^\beta\vlift{\tilde{E}_a}(\tilde{L}).
\]
Here $\tilde{L}_c$ is the restriction of $\tilde{L}$ to $\C$:\ but
this is clearly $L_c$. Now
$\vlift{\tilde{E}_a}(\tilde{L})=\tilde{p}_a=-g_{ab}v^b=0$ on $\C$. So
$\tilde{\Gamma}^0$ satisfies
\[
\tilde{\Gamma}^0(\vlift{X}_\alpha(L_c))-\clift{\bar{X}_\alpha}(L_c)=0.
\]
Since by assumption $L$ is regular with respect to $\D$, these
equations uniquely determine $\tilde{\Gamma}^0$, which is of
second-order type.

On the other hand, the equations determining the nonholonomic dynamics are
\[
\Gamma(\vlift{X_\alpha}(L_c))-\clift{\bar{X}_\alpha}(L_c)=
-{R}^a_{\alpha\beta}v^\beta p_a.
\]
We see that the two agree if and only if ${R}^a_{\alpha\beta}v^\beta
p_a=0$, which is of course the condition for weak consistency, and in
fact as we pointed out earlier for strong consistency as well. \end{proof}
 This
extends further results of \cite{FandB} on Abelian Chaplygin
systems to Chaplygin systems in general.

We pointed out above that $\phi_*\tilde{\Gamma}^0$ provides us with a
solution to the vakonomic problem, whether or not consistency holds.
There are further points of interest about the vakonomic problem in
this special case which we now discuss.

First, we can extend the action of $G$ from $Q$ to $T\D^0$
in a manner modelled on the action on the multipliers discussed
earlier, as follows.  Recall that at the beginning of the previous
section we said that the $\mu_a$ are really the components of a
covector $\mu$ in $\D^0\subset T^*Q$ with respect to the coframe that
is dual to the chosen frame $\{X_\alpha, {\tilde E}_a\}$.  We can also
interpret $\mu_a$ therefore as a function on $T^*Q$.  In fact each
vector field $X$ on $Q$ defines a linear function $\mu_X = \mu_i X^i$
on $T^*Q$, and $\mu_a$ is the function that corresponds to the vector
field ${\tilde E}_a$.  On the other hand, each vector field $W$ on $Q$
can be lifted to a vector field $W^{(1)}$ on $T^*Q$ (another complete
lift:\ see \cite{CP}) and the relation $W^{(1)} (\mu_X) =
\mu_{([W,X])}$ holds.  If we take $W= {\tilde E}_a$ and $X={\tilde
E}_b$, we get
\[
\tilde{E}_a^{(1)}(\mu_b)=-C_{ab}^c\mu_c.
\]
Now the action of $G$ on $Q$ extends to an action on $Q\times \g^* =
\D^\circ$ and therefore also induces an action on $T\D^\circ$.  The
infinitesimal generators of this last action are simply the vector
fields
\[
e_a = \clift{\tilde E}_a + \tilde{E}_a^{(1)}.
\]
The extended Lagrangian $\hat{L}=L-\mu_av^a$, which is a function on
$T\D^0$, is invariant under this action:
\[
e_a(\hat{L})=\clift{\tilde{E}_a}(L) - \tilde{E}_a^{(1)} ( \mu_b) -
\mu_b \clift{\tilde E}_a( v^b)= -\mu_cC^c_{ab}v^b+\mu_cC^c_{ab}v^b =0.
\]
The corresponding momentum components are just $\hat{p}_a=p_a-\mu_a$
(so far as this calculation is concerned, $\mu_a$ is a coordinate on
the base).  We thus obtain by a different method a result that we
pointed out earlier, namely that every Euler-Lagrange field of
$\hat{L}$ is tangent to the level sets of the functions $\hat{p}_a$.
The zero level set of $\hat{p}_a$, namely $\mu_a=p_a$, is of course
the section $\phi$ of $\C\times\R^{n-m}\to\C$ we are using.  Not only
is $\hat{\Gamma}$ tangent to this, so also is $e_a$:
\[
e_a(p_b-\mu_b)=-C_{ab}^c(p_c-\mu_c).
\]

The Euler-Lagrange equations for the vakonomic problem may be written
\[
\hat{\Gamma}(\vlift{X_\alpha}(L))-\clift{\bar{X}_\alpha}(L)
=\hat{p}_a{R}^a_{\alpha\beta}v^\beta,\quad
\hat{\Gamma}(\hat{p}_a)=0.
\]
Set $\hat{\Gamma}^0=\phi_*\tilde{\Gamma}^0$.  Then since
$\tilde{\Gamma}^0(\vlift{X}_\alpha(L))-\clift{\bar{X}_\alpha}(L) =0$
and $\hat{\Gamma}^0$ is tangent to the zero level of $\hat{p}_a$, we
have
\[
\hat{\Gamma}^0(\vlift{X_\alpha}(L))-\clift{\bar{X}_\alpha}(L)
=0,\quad
\hat{\Gamma}^0(\hat{p}_a)=0,
\]
which shows explicitly that $\hat{\Gamma}^0$ satisfies the
Euler-Lagrange equations of $\hat{L}$ on $\im(\phi)$.

It may not have escaped the notice of the reader that $\tilde{L}$ is
actually the Routhian of $L$ (see \cite{Routh,MRS}).  This does not
seem to be of any significance --- we are not interested in the
Euler-Lagrange field of $L$ or the level sets of $p_a$.  However,
there are Routhian overtones to the story.  In the first place, we may
consider the Routhian of $\tilde{L}$ itself.
This is given by
 \[
{\tilde L} - {\tilde p_a} v^a = L - p_av^a + g_{ab}v^av^b.
\]
(Observe that if $L=T-V$ is of mechanical type with $g_{a\alpha}=0$,
this Routhian is simply $L$.)  Its restriction to the level set
$v^a=0$ is simply the constrained Lagrangian $L_c:\C \to \R$.
Furthermore, the so-called Routh procedure, as described in
\cite{Routh}, gives an alternative way of deriving the equations for
$\tilde{\Gamma}^0$, which are in fact the generalized Routh equations
for the Routhian of $\tilde{L}$ on $\tilde{p}_a=0$.

Secondly, we may compute the Routhian of $\hat{L}$, which is
\[
\hat{L}-\hat{p}_av^a=L-\mu_av^a-(p_a-\mu_a)v^a=L-p_av^a=\tilde{L}.
\]
This is a function on $TQ$, in other words it is independent of
$\mu_a$, which itself is unusual and interesting.  Once again one can
use the Routh procedure to derive the equations for $\hat{\Gamma}$
restricted to any level set of $\hat{p}_a$.  These will be expressed
in terms of $\tilde{L}$, which reveals again the close relationship
between this Lagrangian and the vakonomic problem.

\section{Examples}

\subsection{A class of nonholonomic systems with constraints of a special form}

As a first application we look at systems on $\R^{k+2}$, with
coordinates $(q_1,q_2,q_a)$ and corresponding natural fibre
coordinates $(u_1,u_2,u_a)$, $a=3,4,\ldots,k+2$, where the constraints
take the rather special form
\[
{u}_a + \Delta_a(q_1) {u}_2 = 0
\]
for some functions $\Delta_a$, and where the Lagrangian is of the type
of a Euclidean metric
\[
L =\onehalf\left(I_1 {u}_1^2 + I_2 {u}_2^2 + \sum_a I_a {u}_a^2\right).
\]
(We shall not use Einstein's convention for sums over the index $a$ in
this section.)  This class of nonholonomic systems has been studied in
\cite{BFM}, where the authors investigated, using methods different
from ours, whether there is a Lagrangian whose Euler-Lagrange field
coincides with the nonholonomic dynamics when restricted to the
constraint submanifold. Evidently our Corollary~2 is relevant to this
question.  The class includes important classical examples such as the
nonholonomic particle, the vertically rolling disk, the knife edge on
a horizontal plane, the mobile robot with a fixed orientation, etc.

The frame
\[
\left\{\fpd{}{q_1},\fpd{}{q_2}-\sum_a \Delta_a\fpd{}{q_a}\right\}
\]
spans $\D$; it can be completed to a total frame by adding the
vector fields $\partial/\partial {q_a}$.  These last vector fields are
in fact the infinitesimal generators of the $\R^{k}$-action given by
$((\varepsilon_a),(q_1,q_2,q_a))\mapsto (q_1,q_2,q_a+ \varepsilon_a)$,
under which both the Lagrangian and the constraints are invariant
(they are even invariant under an $\R^{k+1}$-action).  The systems are
therefore of Chaplygin type. We write
\[
X_1=\vf{q_1},\quad X_2=\vf{q_2}-\sum_a\Delta_a\vf{q_a},\quad
X_a=\widetilde{E}_a=\vf{q_a}.
\]
In terms of the corresponding quasi-velocities $v_i$ the restriction of
the Lagrangian to the constraint submanifold $\C:v_a=0$ is just
\[
L_c=\onehalf\left(I_1v_1^2 +
\left(I_2+\sum_aI_a\Delta_a^2\right)v_2^2\right).
\]
The ideal candidate for a section to check strong consistency is given
by $\phi_a = p_a= \vlift{X}_a(L) = -I_a\Delta_av_2$ on $\C$ (no sum
over $a$ here).
Since the only non-vanishing bracket of vector fields
in the above frame is
\[
[X_1,X_2] =-\sum_a \Delta'_aX_a,
\]
the condition $\sum_a \phi_aR^a_{\alpha\beta}v^\beta = 0$ becomes
\[
\left(\sum_a I_a \Delta_a\Delta'_a\right)v_2^2= 0 \quad \mbox{and} \quad
\left(\sum_a I_a \Delta_a\Delta'_a\right)v_1v_2=0
\]
on $\C$.  Given that, in general, $v_1,v_2\neq 0$, this condition is
simply $\sum_a I_a \Delta_a \Delta'_a=0$. This last equation is the
condition for the function $I_2+\sum_aI_a\Delta_a^2$, the coefficient
of $v_2^2$ in $L_c$, to be constant; that is to say, it is the
condition for $L_c$ to have constant coefficients when expressed in
terms of the quasi-velocities.  The condition has another geometric
interpretation.  The function
\[
N(q_1) = \frac{1}{\sqrt{I_2+ \sum_a I_a \Delta^2_a(q_1)}}
\]
is the invariant measure density of the above nonholonomic system
(see \cite{BFM}); our condition is therefore the condition for the
system to admit a constant invariant measure density.

From Corollary~\ref{cor1} we can conclude that if $N$ is a constant,
the nonholonomic field $\Gamma$ is the restriction to $\C$ of an
Euler-Lagrange field of the Lagrangian $\tilde{L}=L-\sum_ap_av_a$,
which turns out to be
\[
\onehalf\left(I_1v_1^2+N^{-2}v_2^2-\sum_aI_av_a^2\right).
\]
In terms of the original variables we have
\[
\tilde{L} = \onehalf \left(I_1 {u}_1^2 + I_2 {u}_2^2 - \sum_a
I_a {u}_a^2 \right) - \sum_a \Delta_a I_a {u}_a {u}_2.
\]
This is exactly the statement (in the current terminology) in
Proposition~2 of \cite{BFM}.

Let us come to some details for a few specific examples.  The
Lagrangian of the nonholonomic particle is $L=\onehalf(u_1^2+ u_2^2 +
u_3^2)$ and its constraint is $u_3 + q_1 u_2 = 0$.  One easily
verifies that the nonholonomic vector field is given by
\[
\Gamma = v_1\clift{X_1}+v_2\clift{X_2}
-\frac{q_1v_1v_2}{1+q_1^2}\vlift{X_2}.
\]
(For this example, it so happens that there is no term in
$\vlift{X_1}$.) The vakonomic fields $\hat\Gamma$
take the form
\begin{align*}
\hat\Gamma&=v_1\clift{X_1}+v_2\clift{X_2}\\
&\phantom{=}\mbox{}-\mu q_2\vlift{X_1} +(\mu
v_1+q_1A)\vlift{X_2}+\left((1+q_1^2)A+v_1v_2-\mu
q_1v_1\right)\vlift{X_3}+A\vf{\mu},
\end{align*}
where $A$ is arbitrary.  The field for which $\Gamma_{\C}$ is tangent
to $\C$ (so that the coefficient of $\vlift{X_3}$ vanishes) is
the one where
\[
A=\frac{\mu q_1v_1 - v_1v_2}{1+q_1^2}.
\]
Obviously, the invariant measure density is not a constant, so there
is no strong consistency on $\mu=p$. On the other hand, $\Gamma_\C$
and $\Gamma$ evidently coincide for $\mu=0$, which illustrates our
observation that weak consistency always holds.

The vertically rolling disk also belongs to the class.  Let $R$ be the
radius of the disk.  If the triple $(x,y,z=R)$ stands for the
coordinates of its centre of mass, $\varphi$ for its angle with the
$(x,z)$-plane and $\theta$ for the angle of a fixed line on the disk
and a vertical line, then the nonholonomic constraints are of the form
$u_x = (R\cos\varphi) u_\theta$ and $u_y = (R \sin\varphi) u_\theta$.
The Lagrangian of the disk is $ L = \onehalf M (u_x^2 + u_y^2) +
\onehalf I u_\theta^2 + \onehalf J u_\varphi^2,$ where $I$ and $ J$
are the moments of inertia and $M$ is the total mass of the disk.  The
identification with the notations above is $(q_1,q_2;q_a)=
(\varphi,\theta;x,y)$.  It is easy to see that $N$ is a constant here
and that the nonholonomic and the vakonomic equations are consistent.
For a detailed calculation of those equations, see e.g.\ \cite{Bloch}.

\subsection{A two-wheeled carriage}

Consider a two-wheeled carriage which can move on a horizontal plane
in the direction in which it points and which can spin around a
vertical axis; the wheels roll without slipping over the plane. This
object is sometimes called a planar mobile robot, but we do not use
this description since we think that the word `robot' should be
reserved for devices that are subject to controls of some kind.


\begin{center}
\includegraphics[scale=0.8]{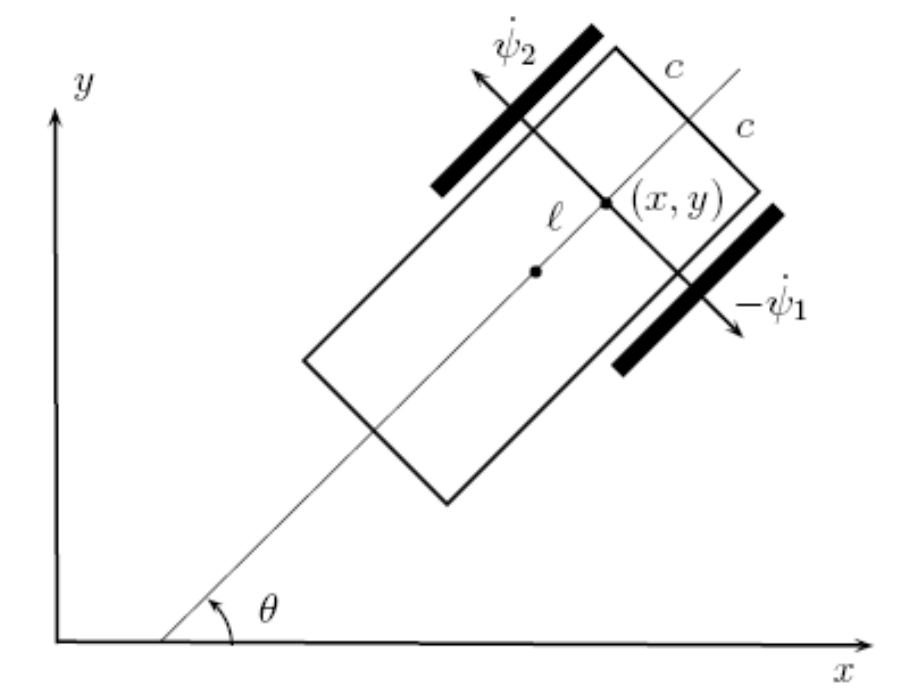}
\end{center}

The angle of rotation about the vertical is denoted by $\theta$, and
the positions of the wheels are characterized by angles $\psi_1$
and $\psi_2$.  We shall generally assume that the intersection point
$(x,y)$ of the horizontal symmetry axis of the carriage with the line
which connects the two wheels is not necessarily the centre of mass of
the system, but lies at a distance $\ell$ from it.

The configuration space of the system is $S^1 \times S^1 \times
SE(2)$, with coordinates $(\psi_1,\psi_2,x,y,\theta)$.
The Lagrangian is
\[
L= \onehalf m (u_x^2 + u_y^2) +
m_0 \ell {u_\theta}((\cos\theta) u_y - (\sin\theta) u_x)
+\onehalf J {u_\theta}^2 + \onehalf J_2 (u_{\psi_1}^2 + u_{\psi_2}^2 )
\]
where $m_0$ is the mass of the body, $m=m_0+2m_1$ is the mass of the
complete system, $J$ is the moment of inertia of the whole system
around a vertical axis through $(x,y)$ and $J_2$ is the axial moment
of inertia for the wheel.  The constraints are
\[
u_x = -\frac{R}{2} \cos\theta (u_{\psi_1} + u_{\psi_2}),\quad
u_y = -\frac{R}{2} \sin\theta (u_{\psi_1} + u_{\psi_2}), \quad
u_\theta = \frac{R}{2c}  (u_{\psi_2} - u_{\psi_1}).
\]
Here $R$ is the radius of a wheel, and $c$ is half the length of the
axle.  This example, including the formulae for the Lagrangian and the
constraints, can be found in the classic reference \cite{Neimark}, and
in many more texts, see e.g.\ \cite{Favretti,FandB,KellyMurray} (the
case where $\ell=0$) and \cite{Frans,CDMM,DLMDD,Koiller} (the general
case).

The distribution $\D$ is spanned by the vector fields
\begin{align*}
X_1 &= \fpd{}{\psi_1} -
\frac{R}{2}\left(\cos\theta \fpd{}{x}+\sin\theta\fpd{}{y} +
\frac{1}{c}\fpd{}{\theta}\right), \\
X_2 &= \fpd{}{\psi_2} -
\frac{R}{2}\left(\cos\theta \fpd{}{x}+\sin\theta\fpd{}{y} -
\frac{1}{c}\fpd{}{\theta}\right).
\end{align*}
The Lagrangian and the constraints are both invariant under the
(usual) $SE(2)$-action and this is an example of a Chaplygin
system.  By adding the fundamental vector fields
\[
X_3=\widetilde{E}_3=\fpd{}{x}, \quad
X_4=\widetilde{E}_4=\fpd{}{y}, \quad
X_5=\widetilde{E}_5=\fpd{}{\theta} - y \fpd{}{x} + x\fpd{}{y},
\]
we get a full frame.  We have $v_1 = u_{\psi_1}$,
$v_2 = u_{\psi_2}$, and the constraints are simply $v_3 =
v_4 =v_5 = 0$, where
\begin{align*}
u_x&=-\onehalf R\cos\theta(v_1+v_2)+v_3-yv_5,\\
u_y&=-\onehalf R\sin\theta(v_1+v_2)+v_4+xv_5,\\
u_\theta&=-\frac{R}{2c}(v_1-v_2)+v_5.
\end{align*}
We have written the quasi-velocities implicitly like this to make it
easier to calculate the expression for $L$ in terms of them.  We don't
in fact need the full expression, so we shan't write it down.  First,
we give the expression for $L_c$:
\[
L_c=\onehalf\left(\frac{R^2}{4c^2}(J+mc^2)+J_2\right)(v_1^2+v_2^2)
-\frac{R^2}{4c^2}(J-mc^2)v_1v_2.
\]
Note that the coefficients are constants.  So far as the full
Lagrangian is concerned, it will shortly become apparent that we
need only those additional terms which are linear in $v_3$ and $v_4$
with coefficients involving only $v_1$ and $v_2$:\ we have
\[
L=L_c-\onehalf mR(v_1+v_2)(v_3\cos\theta+v_4\sin\theta)
+\frac{m_0\ell R}{2c}(v_1-v_2)(v_3\sin\theta-v_4\cos\theta)+\ldots
\]
We proceed to the calculation of the nonholonomic dynamical vector
field. For this we need the bracket of $X_1$ and $X_2$:
\[
[X_1,X_2]=\frac{R^2}{2c}\left((\sin\theta)X_3-(\cos\theta)X_4\right),
\]
from which we obtain the following formulae for derivatives of
quasi-velocities:
\begin{align*}
\clift{X_1}(v_3)&=-\frac{R^2}{2c}(\sin\theta)v_2,\quad
\clift{X_1}(v_4)=\frac{R^2}{2c}(\cos\theta)v_2,\\
\clift{X_2}(v_3)&=\frac{R^2}{2c}(\sin\theta)v_1,\phantom{-}\quad
\clift{X_2}(v_4)=-\frac{R^2}{2c}(\cos\theta)v_1;
\end{align*}
the other derivatives of quasi-velocities along $\clift{X_\alpha}$,
$\alpha=1,2$, are zero (in particular, those of $v_\alpha$).  The
terms not written explicitly in the formula for $L$ above will give
zero when we calculate $\clift{X_\alpha}(L)$ and set $v_3=v_4=v_5=0$.
We have
\[
\clift{X_1}(L)=-\frac{m_0\ell R^3}{4c^2}(v_1-v_2)v_2,\quad
\clift{X_2}(L)=\frac{m_0\ell R^3}{4c^2}(v_1-v_2)v_1.
\]
Recall that $v^\alpha\clift{X_\alpha}$ is tangent to $\C$, and note
that in this case $v^\alpha\clift{X_\alpha}(v^\beta)=0$. For
convenience we write
\[
L_c=\onehalf P(v_1^2+v_2^2)-Qv_1v_2,
\quad P=\frac{R^2}{4c^2}(J+mc^2)+J_2,
\quad Q=\frac{R^2}{4c^2}(J-mc^2).
\]
With
$\Gamma=v_1\clift{X_1}+v_2\clift{X_2}+\Gamma_1\vlift{X_1}+\Gamma_2\vlift{X_2}$
we have
\begin{align*}
\Gamma(\vlift{X_1}(L)-\clift{X_1}(L)&=
P\Gamma_1-Q\Gamma_2+K(v_1-v_2)v_2=0\\
\Gamma(\vlift{X_2}(L)-\clift{X_2}(L)&=
-Q\Gamma_1+P\Gamma_2-K(v_1-v_2)v_1=0,
\end{align*}
where we have written $K$ for $m_0\ell R^3/4c^2$.  Then by simple algebra
\[
\Gamma_1=\frac{K}{P^2-Q^2}(v_1-v_2)(Qv_1-Pv_2),\quad
\Gamma_2=\frac{K}{P^2-Q^2}(v_1-v_2)(Pv_1-Qv_2).
\]
This agrees with the expression of the dynamics in \cite{Koiller}, and
corrects an unfortunate misprint in \cite{DLMDD}.  Notice that $K=0$
when $\ell=0$; in this case the nonholonomic equations say simply that
in any motion $v_1$ and $v_2$ are constants, that is, that the wheels
rotate with constant (in general different) speeds.

In view of the fact that $R^5_{12}=0$, to test for consistency we do
not need to know $p_5$. From the truncated formula for $L$ given
earlier, on $\C$ we have
\begin{align*}
p_3&=-\onehalf mR(v_1+v_2)\cos\theta+\frac{2cK}{R^2}(v_1-v_2)\sin\theta,\\
p_4&=-\onehalf mR(v_1+v_2)\sin\theta-\frac{2cK}{R^2}(v_1-v_2)\cos\theta.
\end{align*}
Thus
\[
p_aR^a_{12}=\frac{R^2}{2c}((\sin\theta)p_3-(\cos\theta)p_4)=
K(v_1-v_2).
\]
It follows that consistency holds (strongly, since this is a
Chaplygin system) with $\phi_a=p_a$ if and only if $\ell=0$.
(Actually, we can choose $\phi_5$ arbitrarily subject to the
condition $\Gamma(\phi_5)=\lambda_5$.)

However, we are not restricted to taking $\phi_a=p_a$. For a
Chaplygin system, the second condition for strong consistency is
$\Gamma(\phi_a)=\lambda_a$; this is certainly satisfied by
$\phi_a=p_a$, but it is also satisfied by $\phi_a=p_a+k_a$ where
$k_a$ is any constant of motion, $\Gamma(k_a)=0$. So we may enquire
whether, when $\ell\neq0$, there are constants of motion $k_3$ and
$k_4$ such that $(p_a+k_a)R^a_{12}=0$, that is, such that
\[
(\sin\theta)k_3-(\cos\theta)k_4=-\frac{2cK}{R^2}(v_1-v_2).
\]
Evidently, we should search for constants of motion which are linear
in $v_1$ and $v_2$ with coefficients linear in $\cos\theta$ and
$\sin\theta$. So first let us set
\[
k=f_1(\theta)v_1+f_2(\theta)v_2.
\]
Now so far as their action on functions of $\theta$ is concerned, both
$X_1$ and $X_2$ are just $\partial/\partial\theta$ up to a constant
factor.  Thus
\begin{align*}
\Gamma(k)&=
(v_1\clift{X_1}+v_2\clift{X_2}+\Gamma_1\vlift{X_1}+\Gamma_2\vlift{X_2})(k)\\
&=-\frac{R}{2c}(v_1-v_2)(f_1'v_1+f_2'v_2)+\Gamma_1f_1+\Gamma_2f_2.
\end{align*}
It will be convenient to set
\[
\Gamma_1=\frac{R}{2c}(v_1-v_2)(\hat{Q}v_1-\hat{P}v_2),\quad
\Gamma_2=\frac{R}{2c}(v_1-v_2)(\hat{P}v_1-\hat{Q}v_2),
\]
which means taking
\[
\hat{P}=\frac{2cKP}{R(P^2-Q^2)},\quad
\hat{Q}=\frac{2cKQ}{R(P^2-Q^2)}.
\]
Then for $k$ to be a constant of motion we require that
\[f_1'v_1+f_2'v_2=(\hat{Q}v_1-\hat{P}v_2)f_1+(\hat{P}v_1-\hat{Q}v_2)f_2,
\]
and this for all $v_1$ and $v_2$. Thus $f_1$ and $f_2$ must satisfy
\[
\left[\begin{array}{c}f_1'\\f_2'\end{array}\right]
=
\left[\begin{array}{cc}\hat{Q}&\hat{P}\\-\hat{P}&-\hat{Q}\end{array}\right]
\left[\begin{array}{c}f_1\\f_2\end{array}\right].
\]
But then
\[
\left[\begin{array}{c}f_1''\\f_2''\end{array}\right]
=
\left[\begin{array}{cc}\hat{Q}^2-\hat{P}^2&0\\0&\hat{Q}^2-\hat{P}^2\end{array}\right]
\left[\begin{array}{c}f_1\\f_2\end{array}\right].
\]
So if we want $f_1$ and $f_2$ to be linear functions of $\cos\theta$
and $\sin\theta$ we had better ensure that $\hat{Q}^2-\hat{P}^2=-1$,
that is, $R^2(P^2-Q^2)=4c^2K^2$, which in terms of the original
parameters can be written
\[
\ell=
\frac{\sqrt{\left(mR^2+2J_2\right)\left(R^2J+2c^2J_2\right)}}{m_0R^2}.
\]
So when the body of the carriage has the special position relative to
the axle which is specified by this value of $\ell$, the system admits
constants of motion of the required form.

Assume this condition is satisfied:\ can we find $k_3$ and $k_4$
such that
\[
(\sin\theta)k_3-(\cos\theta)k_4=-\frac{2cK}{R^2}(v_1-v_2)=
-\hat{K}(v_1-v_2)?
\]
Inspection of this formula suggests taking
\begin{align*}
k_3&=-\hat{K}\sin\theta(v_1-v_2)+H\cos\theta(v_1+v_2)\\
&=(-\hat{K}\sin\theta+H\cos\theta)v_1+(\hat{K}\sin\theta+H\cos\theta)v_2,\\
k_4&=\hat{K}\cos\theta(v_1-v_2)+H\sin\theta(v_1+v_2)\\
&=(\hat{K}\cos\theta+H\sin\theta)v_1+(\hat{K}\cos\theta-H\cos\theta)v_2.
\end{align*}
In order that $k_3$ should actually be a constant of motion its
coefficients must satisfy $f_1'=\hat{Q}f_1+\hat{P}f_2$; this holds,
for $\hat{P}^2-\hat{Q}^2=1$, provided that
\[
H=-\hat{K}(\hat{P}-\hat{Q})=-\frac{2cK}{R^2}(\hat{P}-\hat{Q}).
\]
There are three other conditions to be satisfied if both $k_3$ and
$k_4$ are to be constants of motion; they all lead to this same
formula for $H$ by virtue of the fact that $\hat{P}^2-\hat{Q}^2=1$.
In terms of the original parameters
\[
H=-\frac{mR^2+2J_2}{2R}
\]
for what that's worth.  The main point is that when $\ell$ has the
special value (in terms of the other parameters) given above, the
nonholonomic and vakonomic problems are strongly consistent, via a
section $\phi_a=p_a+k_a$ where the $k_a$ are constants of motion with
$k_3$ and $k_4$ as above; $k_5$ can be chosen arbitrarily (subject to
it being a constant of motion).  In particular, in such a case, as
well as in the case $\ell=0$, the nonholonomic dynamics is the
restriction to $\C$ of the Euler-Lagrange field of a Lagrangian
$\tilde{L}$.

Our solution disagrees with the solution in \cite{CDMM}, but a careful
dimensional analysis easily shows that the solution in \cite{CDMM}
cannot be correct.

\subsubsection*{Acknowledgements}
The first author is a Guest Professor at Ghent University:\ he is
grateful to the Department of Mathematical Physics and Astronomy at
Ghent for its hospitality.  The second author is a Marie Curie Fellow
within the 6th European Community Framework Programme and a
Postdoctoral Fellow of the Research Foundation -- Flanders (FWO).

\end{document}